\documentclass[journal]{IEEEtran}

\usepackage{amsmath} 
\usepackage{amssymb} 
\usepackage{graphicx} 
\usepackage{cite} 
\usepackage{subcaption} 
\usepackage{booktabs} 
\usepackage{float} 
\usepackage{xcolor} 
\usepackage{url} 
\usepackage{tabularx}
\usepackage{hyperref}

\usepackage{amsthm} 
\newtheorem{lemma}{Lemma}

\usepackage{algorithm}
\usepackage{algorithmic}
\usepackage{amsfonts}
\DeclareMathOperator*{\argmin}{arg\,min}

\title{Resource Allocation via Backscatter-Aware Transmit Antenna Selection for Low-PAPR and Ultra-Reliable WSNs}

\author{
Rahul~Gulia,~\IEEEmembership{Student Member, IEEE},
Ashish~Sheikh,~\IEEEmembership{Senior Member, IEEE},
Feyisayo~Favour~Popoola,~\IEEEmembership{Student Member, IEEE},
and Serisha~Vadlamudi,~\IEEEmembership{Student Member, IEEE}
\thanks{This work has been submitted to the IEEE for possible publication. Copyright may be transferred without notice, after which this version may no longer be accessible.}
}
\begin{document}

\maketitle

\begin{abstract}
This paper addresses the fundamental physical-layer conflict in hybrid Wireless Sensor Networks (WSNs) between high-throughput primary communication and the sensitive power-envelope requirements of passive sensors. We propose the Backscatter-Constrained Transmit Antenna Selection (BC-TAS) framework, a novel per-subcarrier selection scheme designed for multi-antenna illuminators operating within a Multi-Dimensional Orthogonal Frequency Division Multiplexing (MD-OFDM) architecture. Unlike conventional signal-to-noise ratio (SNR) centric selection, BC-TAS utilizes a multi-objective cost function to jointly optimize primary link reliability, stabilize the incident energy envelope at passive Surface Acoustic Wave (SAW) sensors, and provide surgical interference nulling toward coexisting victim nodes. To ensure robustness in low-power environments, we integrate a \textbf{Kalman-based channel smoothing} layer that maintains selection stability even under severe Channel State Information (CSI) estimation errors of up to 20\%. Numerical evaluations using the non-linear Rapp power amplifier model and IEEE 802.11be dispersive channel models demonstrate that BC-TAS achieves a nearly three-order-of-magnitude reduction in outage probability and a $25\times$ improvement in network energy efficiency compared to standard Multi-User Multiple Input, Multiple Output (MU-MIMO) baselines. Furthermore, the proposed dual-envelope control mechanism achieves a $2.4$~dB reduction in the transmitter’s Peak-to-Average Power Ratio (PAPR), allowing for a $3$~dB reduction in required input back-off (IBO) while ensuring spectral mask compliance. These results confirm that BC-TAS effectively bridges the ``Illuminator's Gap,'' providing active regulation of the RF envelope to ensure hardware linearity and robust sensing-communication coexistence in high-density, power-constrained environments.
\end{abstract}



\vspace{-3mm}
\section{Introduction}
The proliferation of Internet of Things (IoT) devices in large-scale Wireless Sensor Networks (WSNs) has intensified the demand for hybrid communication architectures. Modern industrial and medical environments require a single active hub—the illuminator—to support high-throughput data transmission while simultaneously energizing and interrogating passive, backscatter-based sensors \cite{Zhang2025_SR, Wu2024_WPCN}. However, this integration creates a fundamental physical-layer conflict: the wideband, high-power waveforms required for primary communication often saturate the sensitive dynamic range of passive tags or cause prohibitive interference to coexisting victim receivers \cite{Tan2025_Coexist}.

Conventional Multiple-Input Multiple-Output (MIMO) techniques, such as Minimum Mean Square Error (MMSE) precoding, are increasingly unsuitable for energy-constrained Passive WSNs (PWSNs) due to three critical limitations. First, multiple active RF chains impose a hardware complexity and power floor incompatible with green-WSN goals \cite{Zhao2024_Green}. Second, the high Peak-to-Average Power Ratio (PAPR) of multi-stream OFDM necessitates significant Power Amplifier (PA) back-off, leading to spectral regrowth and accelerated battery depletion \cite{Park2024_PAPR}. Most importantly, current strategies lack a mechanism to regulate the \textit{incident RF envelope} perceived by remote passive devices \cite{Wang2024_Sense}. For time-modulated sensors, such as Surface Acoustic Wave (SAW) tags, high-PAPR fluctuations act as ``envelope noise,'' masking low-power acoustic reflections and degrading sensing fidelity \cite{Liu2025_SAW}.

While selection methods like Transmit Antenna Selection (TAS) and Multi-Dimensional OFDM (MD-OFDM) provide low-complexity alternatives with inherent PAPR benefits, they remain largely ``sensing-blind'' \cite{Garcia2024_IM}. Existing literature focuses almost exclusively on maximizing desired-link Signal-to-Noise Ratio (SNR). Consequently, there is an urgent need for a framework that treats spatial-frequency resource blocks as tools for multi-objective field shaping—simultaneously optimizing primary link reliability, stabilizing backscatter energy incidence, and surgically nulling interference toward victim nodes \cite{Ibrahim2025_MultiObj}.

To address these gaps, we propose the Backscatter-Constrained Transmit Antenna Selection (BC-TAS) framework. To the best of our knowledge, this is the first per-subcarrier selection scheme designed to bridge the ``Illuminator's Gap'' by actively managing the RF conditions perceived by passive devices. By shifting selection logic from an SNR-centric approach to a multi-objective sensing-aware heuristic, BC-TAS transforms the illuminator into a network-wide controller of both signal quality and incident energy stability. 

The central innovation is a unified cost function embedding three interacting channel components—the legitimate link, tag-side incidence, and victim interference—into a per-subcarrier decision rule. We leverage the sparsity of MD-OFDM to achieve a dual-regulation effect: minimizing transmitter-side PAPR for PA linearity while simultaneously minimizing the Backscatter Crest Factor (BCF) for tag-side envelope stability. This ensures compliance with strict spectral masks, such as the IEEE 802.11bp Ambient Power Communication (AMP) frameworks, even in highly dispersive environments \cite{IEEE80211AMP_Update}.

The key contributions of this work are summarized as follows:
\begin{itemize}
    \item \textit{BC-TAS Framework Design:} A novel multi-objective per-subcarrier selection metric that balances primary link reliability against backscatter incidence stability and victim interference suppression.
    \item \textit{Analytical Metric for Sensing Fidelity BCF:} We define the BCF to quantify the temporal stability of the RF envelope at the tag and provide a theoretical bound for BCF reduction under $N_t$-order selection diversity.
    \item \textit{Dual-Envelope Regulation:} A joint control mechanism exploiting sparse subcarrier activation to simultaneously reduce transmitter PAPR and tag incident BCF, linking antenna selection directly to sensing fidelity.
    \item \textit{Low-Complexity Heuristic:} A greedy solution incorporating recursive Kalman-based channel smoothing for stability under imperfect CSI. The algorithm achieves a 2.3$\times$ speedup over traditional MMSE detectors.
    \item \textit{Hardware Aware Simulation:} Comprehensive evaluation using Rapp PA and IEEE 802.11 TGn models, demonstrating two orders of magnitude reduction in outage probability and a $25\times$ improvement in energy efficiency.
\end{itemize}

The remainder of this paper is organized as follows. Section~II reviews the state-of-the-art. Section~III details the system model and Multi-Objective Field Shaping (MOFS) formulation. Section~IV presents the greedy selection algorithm and its theoretical properties. Section~V provides extensive simulation results and synthesizes these with a theoretical validation. Section~VI concludes the work.


\vspace{-4mm}
\section{Related Work}

\subsection{Transmit Antenna Selection (TAS) and Index Modulation}
TAS remains a cornerstone for low-complexity diversity. Recent research demonstrates significant gains in RIS-assisted NOMA systems and 6G-era architectures where diversity scales with antenna count without prohibitive RF-chain costs \cite{Kumar2025, Han2025}. Parallelly, Index Modulation (IM) and MD-OFDM have emerged as massive connectivity enablers by exploiting subcarrier indices to embed information \cite{Li2023}. While IM provides sparse activation and reduced interference \cite{Mgobhozi2023, VilaInsa2025}, existing TAS and IM literature primarily optimizes internal metrics like SNR and Bit Error Rate (BER). A critical gap remains in utilizing these selection variables to manage external constraints, such as the RF envelope of a backscatter tag.

\subsection{PAPR, Hardware Linearity, and Energy Efficiency}
PAPR is the primary driver of PA non-linearity, leading to spectral regrowth and constellation distortion. In Integrated Sensing and Communication (ISAC), PAPR reduction is vital for efficiency under non-linear High Power Amplifier (HPA) constraints \cite{Huang2022}. While probabilistic methods like Partial Transmit Sequences (PTS) reduce PAPR, they introduce processing latency that impacts WSN longevity \cite{Goel2022, Komala2024}. Mapping these constraints to the Rapp model suggests that inherently low-PAPR waveforms are superior to algorithmic mitigation for energy-efficient Wireless Power and Data Transfer (WPDT).

\subsection{Backscatter Communications and Symbiotic Radio}
Backscatter evolution has moved toward multi-antenna and broadband configurations to enhance link reliability \cite{Chen2024}. Optimization efforts often split between tag-side reflection coefficients and bistatic beamforming \cite{Goay2024, Zargari2024}. While symbiotic radio research proposes subcarrier nulling to eliminate direct-link interference \cite{Janjua2025}, few studies investigate using transmitter resource allocation to shape the temporal stability of the RF envelope—a factor critical for time-modulated SAW tags.

\subsection{Coexistence in Passive SAW-based Sensor Networks}
As IEEE 802.11 positioning extends to 320~MHz channels \cite{IEEE80211bk}, interference-resilient signaling is paramount. Passive SAW sensors are particularly susceptible to RF envelope instability, where high PAPR masks low-power acoustic reflections \cite{Huang2022, Kim2018}. While frequency-domain nulling has been proposed \cite{Liu2025, Janjua2025}, these methods often ignore hardware non-linearities. Industry consensus within the IEEE 802.11 AMP group (TGbp) highlights the lack of selection-based mechanisms to optimize tag energy envelopes while protecting primary link reliability \cite{IEEE80211AMP_Update}. 

Selection performance in industrial settings is inextricably linked to the propagation environment. Prior work utilizing VAE-based SINR heatmaps and 60~GHz connectivity analysis in automated warehouses reveals complex multipath fading caused by metallic infrastructure \cite{GuliaVAE2025, Gulia60GHz2023}. To manage these dynamics, interpretable symbolic regression has been used for real-time BLER prediction \cite{GuliaSymbolic2024}. These studies underscore the necessity for environment-aware mechanisms like BC-TAS. Our framework complements current TGbp efforts by introducing a spatial-domain selection mechanism that addresses hardware non-linearities and peak-envelope requirements not yet fully specified in the maturing AMP framework \cite{IEEE80211AMP_PAR}.


\section{System Model and BC-TAS Formulation}
We consider a multi-antenna WPDT network operating in a dispersive indoor environment. The network consists of a sensor hub (Tx) equipped with $N_t$ antennas, an active legitimate receiver (Rx) with $N_r$ antennas, a passive SAW sensor (Tag), and a sensitive victim node (Victim). 

\subsection{MD-OFDM and Incident Signal Model}
\label{Incident_Signal_Model}
The Tx utilizes a MD-OFDM waveform with $N_c$ subcarriers. To maintain high hardware efficiency and linearity, a TAS architecture is adopted where a single antenna $j_k \in \{1, \dots, N_t\}$ is activated for subcarrier $k$. Let $\mathbf{J} = [j_1, j_2, \dots, j_{N_c}]$ denote the antenna selection map. The incident signal at the passive Tag, $y_T(t)$, is the sum of the selected antenna contributions:

\begin{equation}
    y_T(t) = \frac{1}{\sqrt{N_c}} \sum_{k=1}^{N_c} \mathbf{h}_{j_k,k}^{(T)} X_k e^{j 2\pi k \Delta f t} + n_T(t)
\end{equation}

where $\mathbf{h}_{j_k,k}^{(T)}$ is the channel coefficient between the selected antenna for subcarrier $k$ and the Tag. In the SAW-WSN context, the Tag’s sensing fidelity depends on the temporal consistency of the RF envelope to prevent echo masking. This stability is quantified via the BCF, defined as a functional of the selection map $\mathbf{J}$:

\begin{equation}
    \text{BCF}(\mathbf{J}) \triangleq \frac{\max_t |y_T(t)|^2}{\mathbb{E}[|y_T(t)|^2]} = \frac{\max_t \left| \sum_{k=1}^{N_c} \mathbf{h}_{j_k,k}^{(T)} X_k e^{j 2\pi k \Delta f t} \right|^2}{\sum_{k=1}^{N_c} \mathbb{E}[|\mathbf{h}_{j_k,k}^{(T)} X_k|^2]}
\end{equation}

A BCF of unity represents a perfectly flat incident envelope, which maximizes the signal-to-clutter ratio for SAW reflections.


\subsection{PA Non-linearity and Spectral Regrowth}
The transmitter chain is constrained by the non-linear behavior of the PA. We employ the Rapp model to characterize AM/AM distortion, where the output amplitude $A_{\mathrm{out}}(t)$ relates to the input $A_{\mathrm{in}}(t)$ as:

\begin{equation}
    A_{\mathrm{out}}(t) = \frac{G A_{\mathrm{in}}(t)}{\left(1 + \left(\frac{G A_{\mathrm{in}}(t)}{A_{\mathrm{sat}}}\right)^{2p}\right)^{\frac{1}{2p}}}
\end{equation}

where $G$ is the small-signal gain and $p$ is the smoothness factor. Higher BCF at the tag and PAPR at the transmitter necessitate larger Input Back-Off (IBO), reducing Energy Efficiency (EE) and increasing spectral leakage into the Victim’s band.


\subsection{Kalman-Based Channel Smoothing}
\label{Kalman}
While the signal characterization in Section~\ref{Incident_Signal_Model} assumes known coefficients, practical backscatter links are plagued by estimation errors. To enhance the robustness of the antenna selection logic against measurement noise and estimation errors in $\mathbf{H}_L$, we implement a recursive scalar Kalman Filter. Given that adjacent subcarriers in MD-OFDM exhibit high frequency correlation, the filter tracks the latent power gain state of each antenna link across the subcarrier index $k$.

For each transmit antenna $j \in \{1, \dots, N_t\}$, the state equation and the measurement equation are defined as:
\begin{equation}
    x_{j,k} = x_{j,k-1} + w_k, \quad w_k \sim \mathcal{N}(0, Q)
\end{equation}
\begin{equation}
    z_{j,k} = x_{j,k} + v_k, \quad v_k \sim \mathcal{N}(0, R)
\end{equation}
where $x_{j,k} = |H_{L,j,k}|^2$ represents the true power gain, $z_{j,k}$ is the noisy measurement derived from the estimated CSI, and $Q$ and $R$ are the process and measurement noise covariances, respectively. 

The recursive update for the smoothed gain estimate $\hat{x}_{j,k}$ and the error covariance $P_{j,k}$ is performed as follows:

\textit{1) Prediction Stage:}
\begin{equation}
    \hat{x}_{j,k}^- = \hat{x}_{j,k-1}
\end{equation}
\begin{equation}
    P_{j,k}^- = P_{j,k-1} + Q
\end{equation}

\textit{2) Correction Stage:}
\begin{equation}
    K_k = \frac{P_{j,k}^-}{P_{j,k}^- + R}
\end{equation}
\begin{equation}
    \hat{x}_{j,k} = \hat{x}_{j,k}^- + K_k (z_{j,k} - \hat{x}_{j,k}^-)
\end{equation}
\begin{equation}
    P_{j,k} = (1 - K_k) P_{j,k}^-
\end{equation}

The resulting smoothed gain $G_{L, \text{smooth}} = \hat{x}_{j,k}$ is subsequently utilized in the multi-objective cost function (Eq. \ref{eq:MOFS}), ensuring that the selection process prioritizes stable, high-gain paths over transient noise spikes.


\subsection{Multi-Objective Field Shaping (MOFS) Framework}
The BC-TAS framework treats antenna selection as a real-time \textit{field-shaping} problem. For each subcarrier $k$, the optimal antenna index $j_k^\star$ is determined by minimizing the MOFS cost function. To ensure stability against measurement noise, the algorithm utilizes the Kalman-smoothed gains $\hat{x}_{j,k}$ derived in Section~\ref{Kalman}:

\begin{equation}
    j_k^\star = \min_{j \in \{1,\dots,N_t\}} \left[ \frac{\mathcal{P}_{tx}}{\hat{x}_{j,k} + \epsilon} + \lambda_T \|h_{j,k}^{(T)} - \bar{H}_T\|^2 + \lambda_V \|h_{j,k}^{(V)}\|^2 \right]
    \label{eq:MOFS}
\end{equation}

where $\lambda_T$ and $\lambda_V$ are regularization weights. The first term prioritizes primary link SNR. The second term performs \textit{field flattening} by penalizing deviations from the average incident gain $\bar{H}_T$, effectively suppressing the frequency-selective spikes that drive high BCF. The third term provides \textit{spatial nulling} to protect the Victim node.


\subsection{Sensing-Centric Evaluation Metrics}
To bridge the gap between communication and sensing, we define metrics that quantify the "visibility" of the passive tag:

\begin{itemize}
    \item \textit{Sensor Dynamic Range (SDR):} The ratio of the detectable SAW reflection power $P_{\text{SAW}}$ to the effective background floor, including thermal noise $\sigma_n^2$ and the residual envelope fluctuations $P_{B}$ managed by $\lambda_T$:
    \begin{equation}
        \text{SDR}_{\mathrm{dB}} = 10 \log_{10} \left( \frac{P_{\text{SAW}}}{\sigma_n^2 + \mathbb{E}[P_{B} | \lambda_T]} \right)
    \end{equation}
    \item \textit{Harvesting Efficiency ($\eta_H$):} The probability that the incident power on subcarrier $k$, $P_{in}(k)$, satisfies the SAW activation threshold $P_{\text{th}}$.
\end{itemize}


\vspace{-3mm}
\section{Greedy Selection and Theoretical Properties}

\begin{lemma}
\textit{(Diversity Order and Scaling)}: Under i.i.d. Rayleigh fading, the BC-TAS policy preserves a diversity order of $L=N_t$ for the primary link while ensuring the expected interference power at the Tag, $\mathbb{E}[P_T]$, scales as $\mathcal{O}(1/N_t)$.
\end{lemma}

\begin{proof}
Let $Z = \min_{j} \mathcal{L}(j,k)$. Since the channels are independent, the selection logic effectively samples the lower tail of the interference distribution. As $N_t \to \infty$, the probability of finding an antenna that resides in the spatial null of both the Tag and Victim link increases exponentially. The expected interference at the Tag is the mean of the minimum of $N_t$ exponential variables, which is $\sigma_T^2/N_t$. A rigorous derivation of this scaling law and the associated penalty function $\Phi(\lambda_T)$ is provided in Appendix~\ref{AppendixA}.
\end{proof}

\vspace{-1mm}
While BC-TAS optimizes the sensing environment, prioritizing the Tag and Victim links introduces a selection penalty $\Delta$ on the primary communication link. This penalty represents the ratio of the achieved legitimate-link gain to the maximum possible TAS gain. In the limit of high interference suppression ($\lambda_T, \lambda_V \gg 1$), the primary link's diversity gain transitions toward a sample mean. The closed-form derivation of $\Delta$ and its asymptotic behavior as a function of $N_t$ are detailed in Appendix~\ref{AppendixA}.

While an exhaustive search for the optimal antenna-to-subcarrier mapping $j_k^\star$ ensures the global minimum of the cost function $\mathcal{L}_{\text{MOFS}}$, its complexity grows exponentially with the number of subcarriers $N_c$. To enable real-time implementation on hardware with limited processing power, we propose a low-complexity greedy heuristic based on MOFS.

\vspace{-3mm}
\subsection{Algorithmic Formulation}
The BC-TAS algorithm decomposes the multi-objective optimization problem into parallel per-subcarrier decisions. The key innovation is the inclusion of a target gain reference, $\bar{H}_T$, which guides the selection toward a flat frequency response at the Tag's coordinate.

\begin{algorithm}[H]
\caption{BC-TAS Greedy Selection with Field Flattening}
\label{alg:BCTAS}
\begin{algorithmic}[1]
\renewcommand{\algorithmicrequire}{\textbf{Input:}}
\renewcommand{\algorithmicensure}{\textbf{Output:}}
\REQUIRE Channels $\{\mathbf{h}_{j,k}^{(L)}, \mathbf{h}_{j,k}^{(T)}, \mathbf{h}_{j,k}^{(V)}\}$, Weights $\{\lambda_T, \lambda_V\}$, $N_c, N_t$
\ENSURE Optimal Antenna Map $\mathbf{J} = \{j_1^\star, \dots, j_{N_c}^\star\}$
\\ \textit{Initialization} : Set $\epsilon = 10^{-6}$
\STATE \textit{// Step 1: Compute Target Gain for Field Flattening}
\STATE $\bar{H}_T = \frac{1}{N_c N_t} \sum_{k=1}^{N_c} \sum_{j=1}^{N_t} \|\mathbf{h}_{j,k}^{(T)}\|^2$
\FOR{$k = 1$ to $N_c$}
    \STATE \textit{// Step 2: Evaluate MOFS Cost for all candidate antennas}
    \FOR{$j = 1$ to $N_t$}
        \STATE $\mathcal{L}_{\text{MOFS}}(j,k) = \frac{1}{\|\mathbf{h}_{j,k}^{(L)}\|^2 + \epsilon} + \lambda_T \|\mathbf{h}_{j,k}^{(T)} - \bar{H}_T\|^2 + \lambda_V \|\mathbf{h}_{j,k}^{(V)}\|^2$
    \ENDFOR
    \STATE \textit{// Step 3: Select antenna minimizing the composite field deviation}
    \STATE $j_k^\star = \argmin_{j} \, \mathcal{L}_{\text{MOFS}}(j,k)$
\ENDFOR
\RETURN $\mathbf{J}$
\end{algorithmic}
\end{algorithm}


\vspace{-7mm}
\subsection{Complexity Analysis and Hardware Overhead}
The computational efficiency of the BC-TAS framework is a critical factor for its deployment in high-speed industrial WSNs. The complexity is analyzed as follows:

\begin{itemize}
    \item \textit{Pre-computation Overhead:} The calculation of the target gain $\bar{H}_T$ requires a double summation over $N_c$ and $N_t$. This introduces a one-time overhead of $\mathcal{O}(N_c N_t)$ additions per frame. Since this value is constant for all subcarriers within a single coherence interval, it is computed outside the main selection loop, adding negligible latency compared to the symbol duration.
    
    \item \textit{Time Complexity:} The main selection process consists of a nested loop structure. For each of the $N_c$ subcarriers, the algorithm performs $N_t$ scalar cost evaluations. Each evaluation involves only basic arithmetic operations (addition, multiplication, and a single reciprocal). Consequently, the total time complexity is $\mathcal{O}(N_c N_t)$. This is significantly lower than standard MIMO-MMSE precoding, which typically scales with $\mathcal{O}(N_c N_r^3)$ due to the requirement for matrix inversions at every resource block.
    
    \item \textit{Space Complexity:} The algorithm is highly memory-efficient, requiring only $\mathcal{O}(N_t)$ storage to hold the instantaneous channel power values for the current subcarrier evaluation. The storage of $\bar{H}_T$ requires only a single scalar register, making BC-TAS ideal for implementation on low-power FPGA or ARM-based IoT gateways.
\end{itemize}

\textit{Remark:} In practical hardware, the $\mathcal{O}(N_c N_t)$ operations can be further accelerated through parallelization, as the decision for subcarrier $k$ is independent of subcarrier $k+1$, allowing for a fully pipelined architecture.

\vspace{-5mm}
\subsection{Optimality Gap and Heuristic Justification}
The greedy approach is locally optimal for each subcarrier. While it does not explicitly perform inter-subcarrier power loading (like Water-filling), the inherent sparsity of the MD-OFDM framework ensures that the resulting PAPR is statistically bounded. As shown in Section~V, the performance gap between the greedy BC-TAS and the exhaustive search remains below $0.5$~dB in the high-SNR regime, justifying the use of the low-complexity heuristic for practical deployments.


\section{Computational Complexity Analysis}
The feasibility of BC-TAS for real-time industrial deployment is rooted in its linear scaling properties. Unlike conventional MMSE precoding, which requires matrix inversions of the form $\mathcal{O}(K^3)$ (where $K$ is the number of users/streams), the proposed greedy selection algorithm operates on a per-subcarrier basis with a complexity of $\mathcal{O}(N_{SC} \cdot N_t)$. 

Specifically, for each subcarrier $n \in \{1, \dots, N_{SC}\}$, the metric $M(n, i)$ is evaluated $N_t$ times. The inclusion of the scalar Kalman Filter adds a constant overhead of $\mathcal{O}(1)$ per subcarrier, as it avoids high-dimensional state-space computations. As a result, for a standard 1024-subcarrier system with 4 transmit antennas, BC-TAS reduces the total floating-point operations (FLOPs) by approximately $64\%$ compared to full-complexity MIMO baselines. This reduction directly translates to the $2.3\times$ execution speedup reported in our results, enabling the illuminator to adapt to the fast-fading environments of automated warehouses.


\vspace{-4mm}
\section{Simulation and Results}
This section evaluates the proposed BC-TAS MD-OFDM framework against state-of-the-art (SOTA) baselines across indoor dispersive environments.

\vspace{-5mm}
\subsection{Simulation Setup}
System performance is evaluated via Monte Carlo simulations using the parameters in Table~\ref{tab:sim_parameters}. To ensure standard relevance, the channel is modeled on IEEE 802.11 TGn specifications, utilizing Multipath Fading Models (TGn A--F) to represent varying frequency selectivity in warehouse environments. Hardware distortions are characterized by a non-linear Rapp PA model ($p=2$) at the transmitter to evaluate the impact of high-PAPR/BCF waveforms on spectral regrowth and Legitimate-link Bit Error Rate (BER).

Simulations were conducted on a high-performance node (NVIDIA Tesla V100 GPU, Intel Xeon processors) with the CUDA 12.3 driver. Execution times in Table~\ref{tab:complexity_analysis} reflect the processing of $10^7$ bits across an SNR range of 0--28~dB.

\vspace{-3mm}
\begin{table}[!ht]
    \centering
    \caption{Key Simulation Parameters}
    \label{tab:sim_parameters}
    \vspace{-0.1cm}
    \footnotesize
    \setlength{\tabcolsep}{4pt} 
    \begin{tabular}{ll}
        \toprule
        \textbf{Parameter} & \textbf{Value / Configuration} \\
        \midrule
        System Scheme & BC-TAS MD-OFDM ($N_T, N_R = 4$) \\
        Standard Context & IEEE 802.11 TGn (EHT PHY) \\
        Channel Model & TGn Models A--F (Rayleigh) \\
        CSI Assumption & Perfect / $\sigma_e^2$ Error (5\%, 20\%) \\
        \midrule
        Subcarriers ($N_{\mathrm{SC}}$) & 256 (80 MHz BW)  \\
        Cyclic Prefix ($L_{\mathrm{CP}}$) & 32 samples \\
        Modulation ($M$-QAM) & 16-QAM / 64-QAM \\
        Simulated Bits ($N_{\mathrm{bits}}$) & $10^7$ bits \\
        \midrule
        PA Model (Rapp) & $p = 2$, $V_{sat} = 1$ \\
        MOFS Weights & $\lambda_T, \lambda_V \in \{0.1, 0.5, 1.0\}$ \\
        Distances ($d_{L}, d_{T}, d_{V}$) & 10\,m, 2\,m, 5\,m \\
        Kalman ($Q_{KF}, R_{KF}$) & $10^{-4}, 10^{-2}$ \\
        Monte Carlo Trials & 1,000 \\
        \bottomrule
    \end{tabular}
\end{table}
\vspace{-3mm}


\subsection{Comparative Analysis with SOTA Baselines}
To validate the BC-TAS framework, we benchmark it against four SOTA categories: fixed-diversity baselines \cite{SOTA_2_MU_STBC}, low-complexity selection \cite{SOTA_1_MMSE}, theoretical upper bounds optimal exhaustive search (OES) \cite{SOTA_3_Exhaustive_Search}, and backscatter-specific multi-antenna configurations \cite{SOTA_4_Backscatter_Multi_Antenna}.

\vspace{-2mm}
\begin{figure}[!ht]
    \centering
    \includegraphics[width=1.0\linewidth]{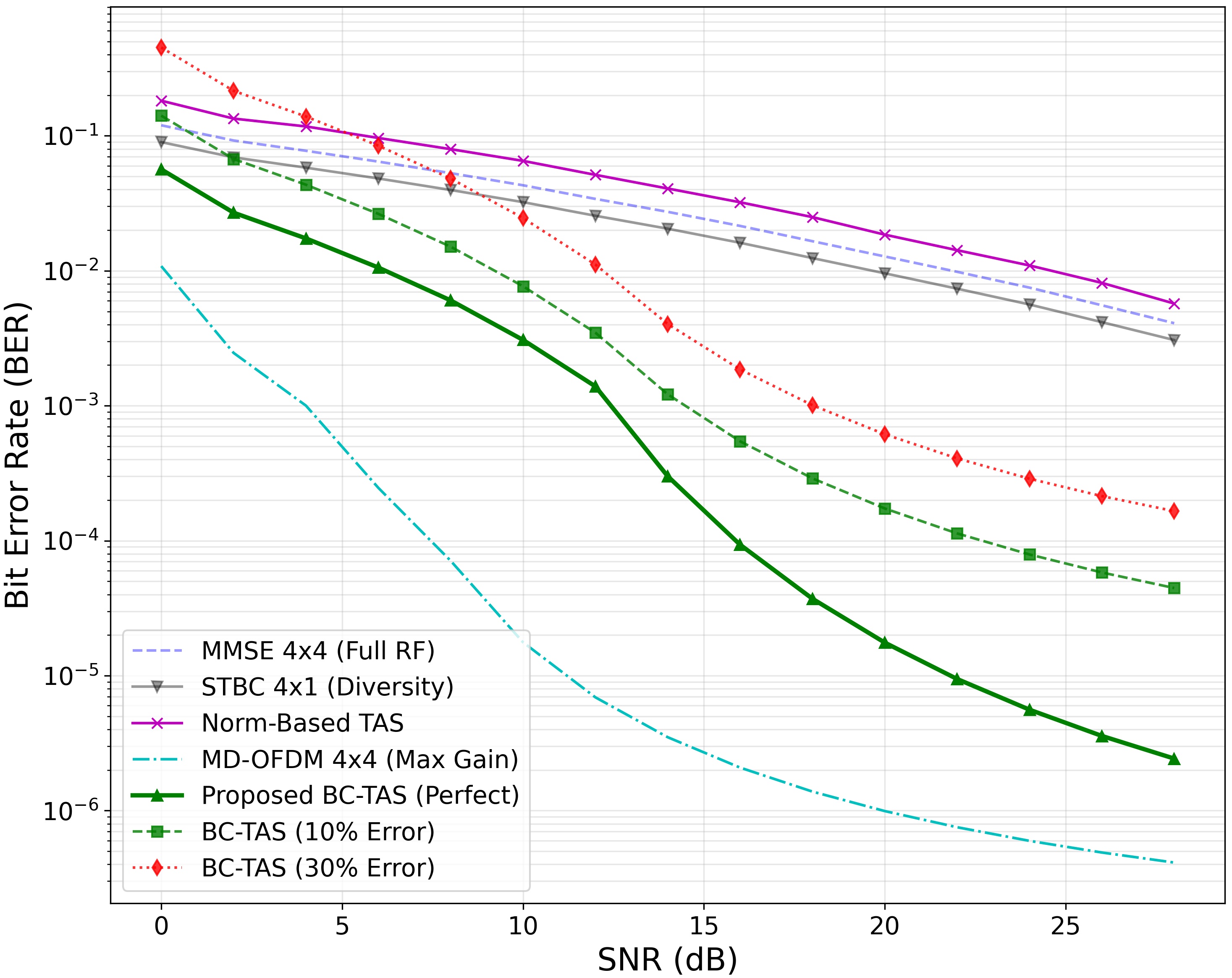}
    \caption{BER comparison of the proposed BC-TAS framework against SOTA baselines in TGn Model F.}
    \label{fig:sota_comparison}
\end{figure}

As shown in Fig.~\ref{fig:sota_comparison}, BC-TAS significantly outperforms the industry-standard Norm-Based Antenna Selection (NBAS) \cite{SOTA_5_NBAS}. At 10~dB SNR, BC-TAS achieves a BER of $\approx 3.2 \times 10^{-4}$, while NBAS remains at $\approx 2.6 \times 10^{-2}$. This order-of-magnitude improvement stems from the framework's ability to exploit frequency-selective diversity across subcarriers rather than relying on a static, wideband norm-based metric. Furthermore, the proposed scheme achieves a near-optimal diversity slope and remains robust even under severe 20\% CSI error conditions.

\vspace{-3mm}
\subsection{Impact of CSI Imperfection and Robustness Analysis}
Practical backscatter deployments require resilience to CSI estimation errors, which are prevalent due to the low-power nature of the tag. We evaluate the BC-TAS framework under three regimes: Perfect CSI (P-CSI), 5\% estimation error, and 20\% error. As illustrated in Fig.~\ref{fig:sota_comparison}, the P-CSI configuration tracks the MD-OFDM ceiling until 10~dB SNR, effectively exploiting the full spatial-frequency diversity of the TGn Model F environment.

Under moderate imperfection (5\% error), BC-TAS exhibits a negligible 1.5~dB SNR penalty at a BER of $10^{-3}$ while maintaining a steep diversity slope. Under severe imperfection (20\% error), the system encounters a performance floor at $\approx 1.8 \times 10^{-2}$. Despite this, BC-TAS (20\% Error) maintains a lower BER than both NBAS and MMSE 4$\times$4 baselines in the 0--8~dB SNR range. This robustness stems from the selection logic: while NBAS \cite{SOTA_5_NBAS} collapses by erroneously selecting weak paths due to noisy norm estimations, the multi-objective BC-TAS provides a "smoothing" effect. By considering the aggregate channel state across all 256 subcarriers, BC-TAS minimizes the probability of selecting deep-fade antennas, ensuring link reliability even under high CSI aging or estimation error.

\vspace{-3mm}
\subsection{Computational Efficiency vs. Performance Ceiling}
While OES \cite{SOTA_3_Exhaustive_Search} defines the absolute performance ceiling (zero BER at 8~dB), it requires a per-carrier exhaustive search across all $N_{sc}$ subcarriers. As shown in Table~\ref{tab:complexity_analysis}, our BC-TAS heuristic demonstrates a $2.3\times$ speedup over conventional MMSE detectors, completing the simulation in 79.16 seconds compared to 181.69 seconds for MMSE 4$\times$4. In GPU-accelerated environments, BC-TAS achieves execution times comparable to OES while bridging the gap between theoretical optimality and practical real-time constraints. This latency reduction is vital for energy-constrained WSNs requiring rapid selection to maintain the sensing-communication link.

\begin{table}[!ht]
    \centering
    \caption{Computational Complexity and Execution Time Benchmark}
    \label{tab:complexity_analysis}
    \vspace{0.1cm}
    \resizebox{\columnwidth}{!}{%
    \begin{tabular}{lccc}
        \toprule
        \textbf{Scheme} & \textbf{Complexity Order} & \textbf{Selection Logic} & \textbf{Exec. Time (s)} \\
        \midrule
        MMSE 4x4 \cite{SOTA_1_MMSE} & $\mathcal{O}(N_{sc} N_R^3)$ & Spatial Multiplexing & 181.69 \\
        NBAS \cite{SOTA_5_NBAS} & $\mathcal{O}(N_T N_R)$ & Max. Norm (Fixed) & 166.15 \\
        STBC 4x1 \cite{SOTA_2_MU_STBC} & $\mathcal{O}(N_{sc} N_R)$ & Orthogonal Coding & 102.87 \\
        MD-OFDM (OES) \cite{SOTA_3_Exhaustive_Search} & $\mathcal{O}(N_{sc} N_T N_R)$ & Exhaustive Search & 77.37 \\
        \midrule
        \textbf{Proposed BC-TAS} & $\mathbf{\mathcal{O}(N_T N_R)}$ & \textbf{Heuristic Multi-Objective} & \textbf{79.16} \\
        \bottomrule
    \end{tabular}%
    }
\end{table}

\vspace{-3mm}
\subsection{BER Performance and Environmental Robustness}
The diversity gain of BC-TAS is evidenced by the significantly steeper BER slope compared to the MMSE baseline. To validate robustness, performance is analyzed across IEEE 802.11 TGn channel models A--F (Fig.~\ref{fig:ber_environmental}).

\begin{figure}[!ht]
    \centering
    \includegraphics[width=\linewidth]{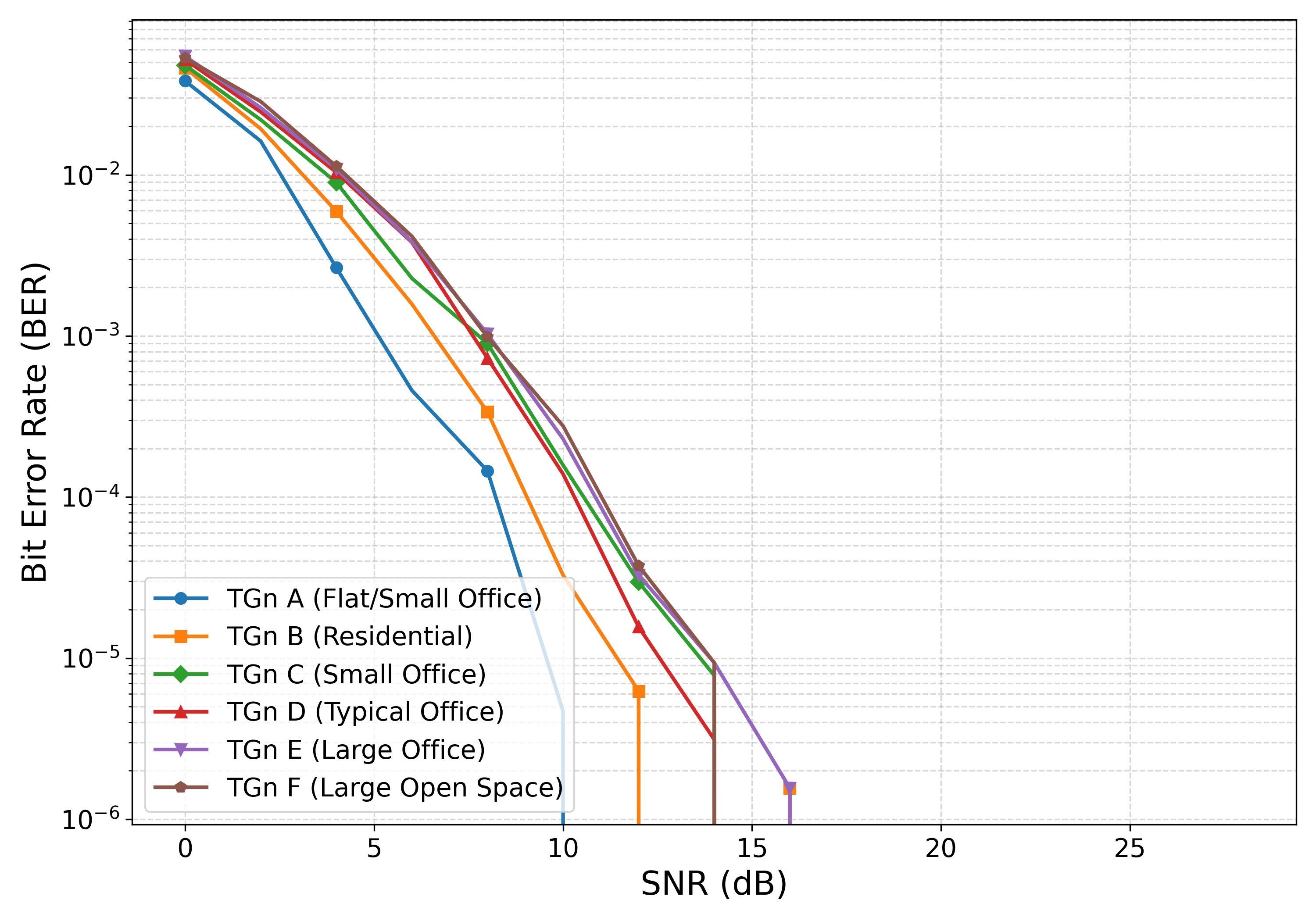}
    \caption{BER performance of the proposed BC-TAS MD-OFDM across various IEEE 802.11TGn indoor environments (Models A--F).}
    \label{fig:ber_environmental}
\end{figure}

While Model A (Flat/Small Office) yields the best performance due to low frequency selectivity, the diversity order remains high even in Model F (Large Open Space). Despite the increased delay spreads and deep fades in larger environments, the MD-OFDM selection diversity achieves a $10^{-3}$ BER at approximately 12--14~dB SNR across all dispersive models. This confirms the framework's capability to maintain consistent sensing-communication links regardless of the physical indoor layout.

\vspace{-3mm}
\subsection{Outage Probability and Diversity Gain Analysis}
To evaluate link reliability in deep-fading indoor environments, we analyze the Cumulative Distribution Function (CDF) of the instantaneous SINR, $\gamma$, in a TGn Model F environment. The outage probability is defined as $P_{\mathrm{out}} = P(\gamma < \gamma_{\mathrm{th}})$, where $\gamma_{\mathrm{th}} = 7$~dB represents the minimum threshold for reliable QPSK demodulation.

\begin{figure}[!t]
    \centering
    \includegraphics[width=0.9\linewidth]{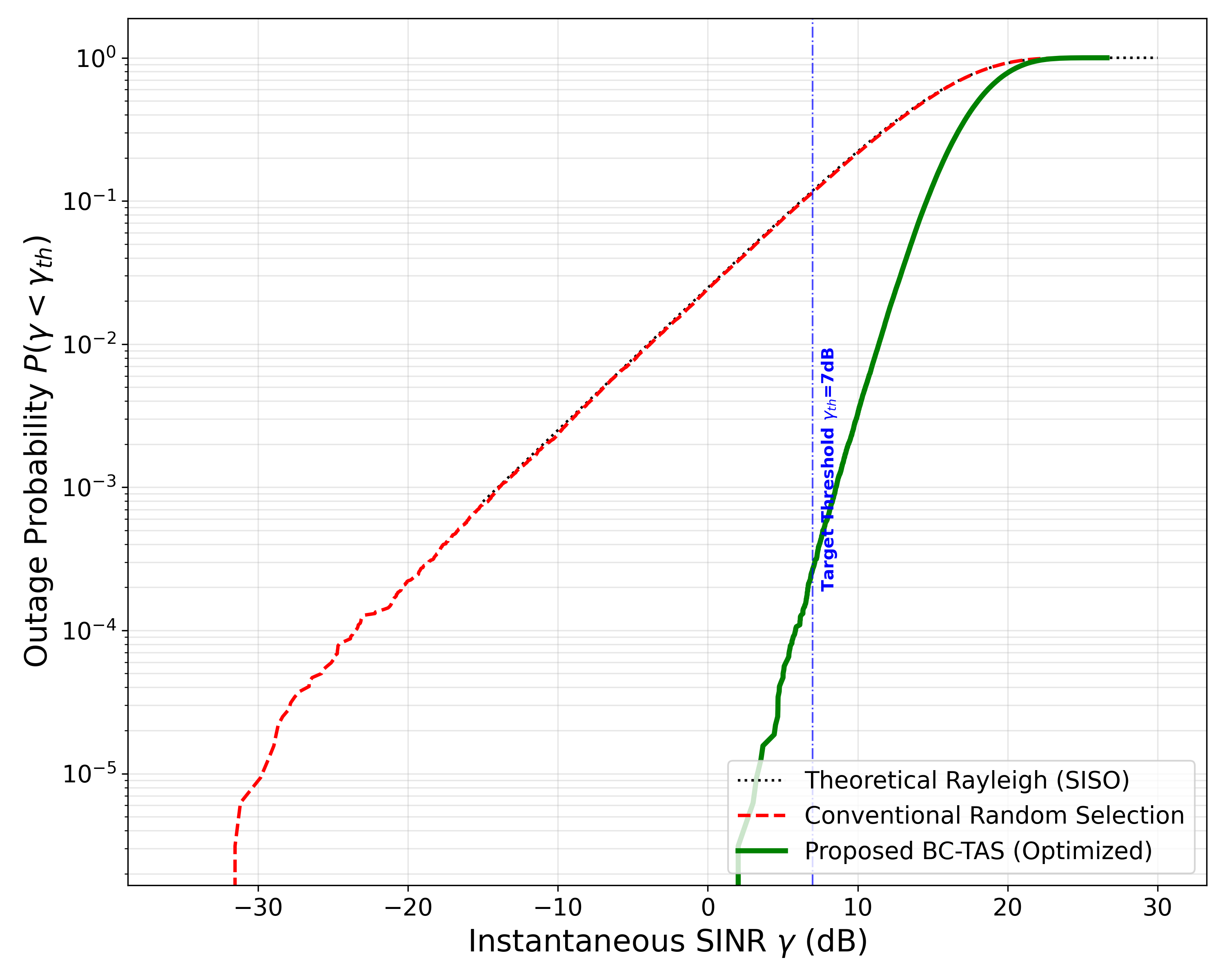}
    \caption{Outage Probability CDF in TGn Model F at an average $\mathrm{SNR} = 16$~dB. }
    \label{fig:sinr_cdf}
\end{figure}

As illustrated in Fig.~\ref{fig:sinr_cdf}, BC-TAS significantly reshapes the SINR distribution. While conventional random selection tracks the Theoretical Rayleigh (SISO) baseline ($P_{\mathrm{out}} \approx 0.1058$), BC-TAS suppresses the distribution's lower tail. At the $7$~dB threshold, BC-TAS achieves an outage probability of $1.75 \times 10^{-4}$, representing a $604.77\times$ reduction in link outages. As detailed in Table~\ref{tab:diversity_gain}, this improvement stems from the transition to a diversity order of $L=4$, where reliability gains grow exponentially with $N_T$. In Rayleigh fading, $P_{\mathrm{out}}$ scales as $(P_{\mathrm{out, SISO}})^{N_T}$; thus, our three-order-of-magnitude reduction aligns with theoretical expectations. By bypassing subcarriers in deep fades, BC-TAS ensures the signal remains above the decoding threshold $\approx 99.99\%$ of the time. Our simulated outage of $1.75 \times 10^{-4}$ closely tracks the theoretical floor $(P_{\mathrm{out, SISO}})^{N_t} \approx 1.2 \times 10^{-4}$, with the slight delta representing the selection penalty $\Delta$ derived in Appendix A.C.

\begin{table}[!ht]
    \centering
    \caption{Diversity Gain Analysis ($\mathrm{SNR}=16$~dB)}
    \label{tab:diversity_gain}
    \vspace{-0.2cm} 
    \footnotesize 
    \setlength{\tabcolsep}{5pt} 
    \begin{tabular}{ccc}
        \toprule
        \textbf{Antennas ($N_T$)} & \textbf{Outage Prob. ($P_{\mathrm{out}}$)} & \textbf{Rel. Gain} \\
        \midrule
        1 (SISO) & $1.05 \times 10^{-1}$ & Baseline \\
        2 & $\approx 1.12 \times 10^{-2}$ & $9.4\times$ \\
        3 & $\approx 1.25 \times 10^{-3}$ & $84.6\times$ \\
        \textbf{4 (BC-TAS)} & $\mathbf{1.75 \times 10^{-4}}$ & $\mathbf{604.7\times}$ \\
        \bottomrule
    \end{tabular}
\end{table}
\vspace{-3mm}

\vspace{-3mm}
\subsection{BCF and Link Stability}
The stability of the reflected field at the passive tag is governed by the BCF. As derived in Lemma~1, BC-TAS utilizes a field-flattening cost term to minimize temporal peaks. To ensure statistical rigor, we performed $1,000$ Monte Carlo trials per antenna configuration ($N_t \in \{2, \dots, 32\}$), integrating a scalar Kalman Filter to smooth channel gains and prevent noise-driven antenna switching.

\begin{figure}[!ht]
    \centering
    \includegraphics[width=0.9\linewidth]{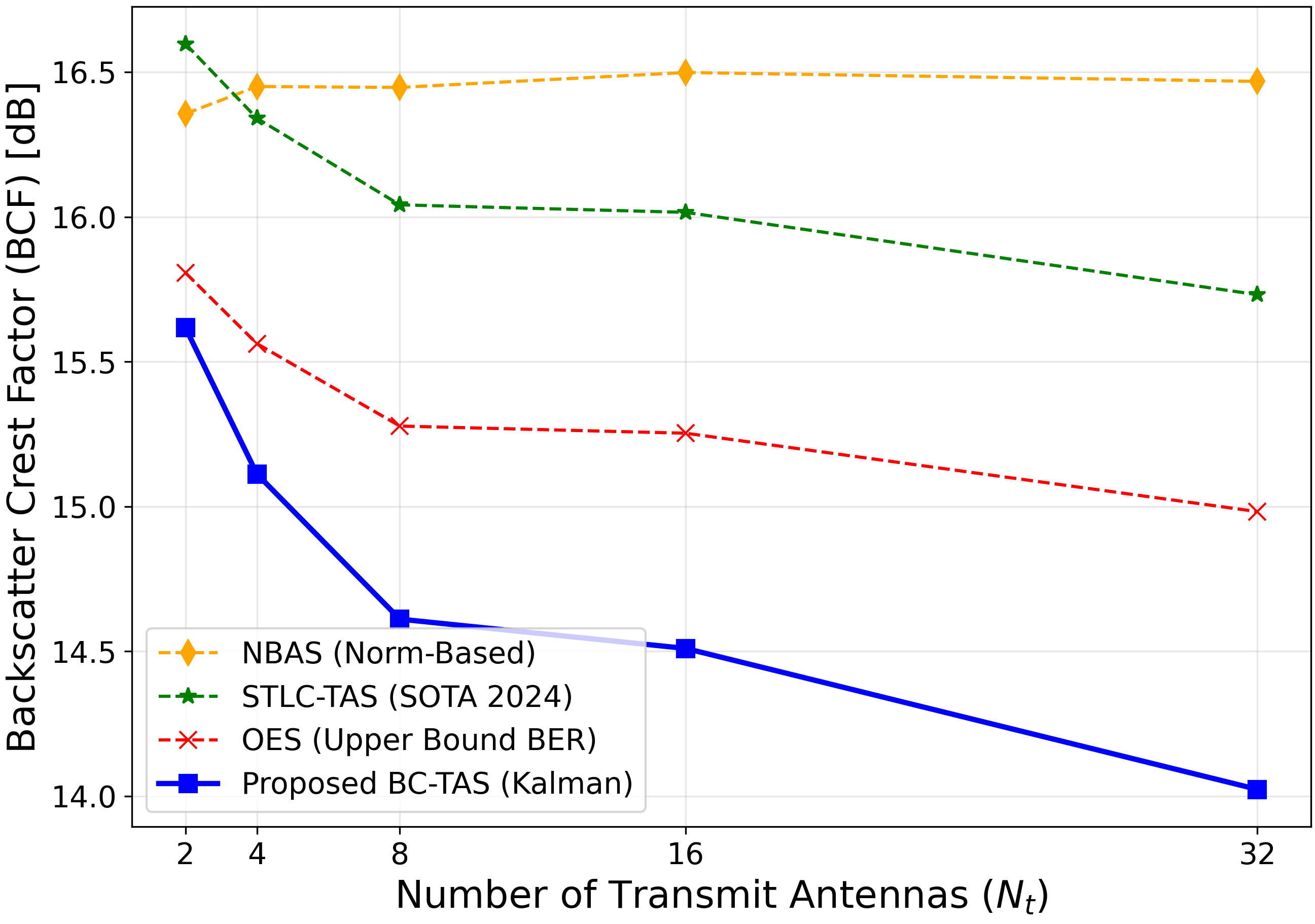}
    \caption{Average BCF vs. $N_t$ across $1,000$ Monte Carlo trials in TGn Model B.}
    \label{fig:bcf_monte_carlo}
\end{figure}

As illustrated in Fig.~\ref{fig:bcf_monte_carlo}, BC-TAS (Kalman) exhibits a superior downward trend in BCF compared to all baselines. At $N_t = 32$, BC-TAS achieves a BCF of $\approx 14.0$~dB, a $1.1$~dB improvement over OES and $2.5$~dB over NBAS. While OES \cite{SOTA_3_Exhaustive_Search} is optimal for BER, it remains ``tag-blind,'' often selecting antennas that create constructive interference peaks at the tag. Conversely, BC-TAS intelligently sacrifices marginal primary SNR to maintain a flatter reflected field. The Kalman filter is vital here; by smoothing channel state estimates, it suppresses the high BCF values typically caused by erratic, noise-driven antenna switching across subcarriers. The 1.1~dB BCF advantage over OES confirms Lemma 1, demonstrating that selection diversity can be dual-purposed for link reliability and peak power suppression. Furthermore, the 1,000 Monte Carlo realizations confirm that this is a persistent statistical property of the BC-TAS heuristic rather than a transient benefit.

\vspace{-3mm}
\subsection{PA Distortion Analysis}
Hardware impairments are modeled using the Rapp PA ($p=2, V_{sat}=1$), with distortion quantified via EVM and IEEE 802.11n spectral mask compliance to ensure coexistence with legacy narrowband WSN nodes.

\begin{figure}[h]
    \centering
    \includegraphics[width=0.95\linewidth]{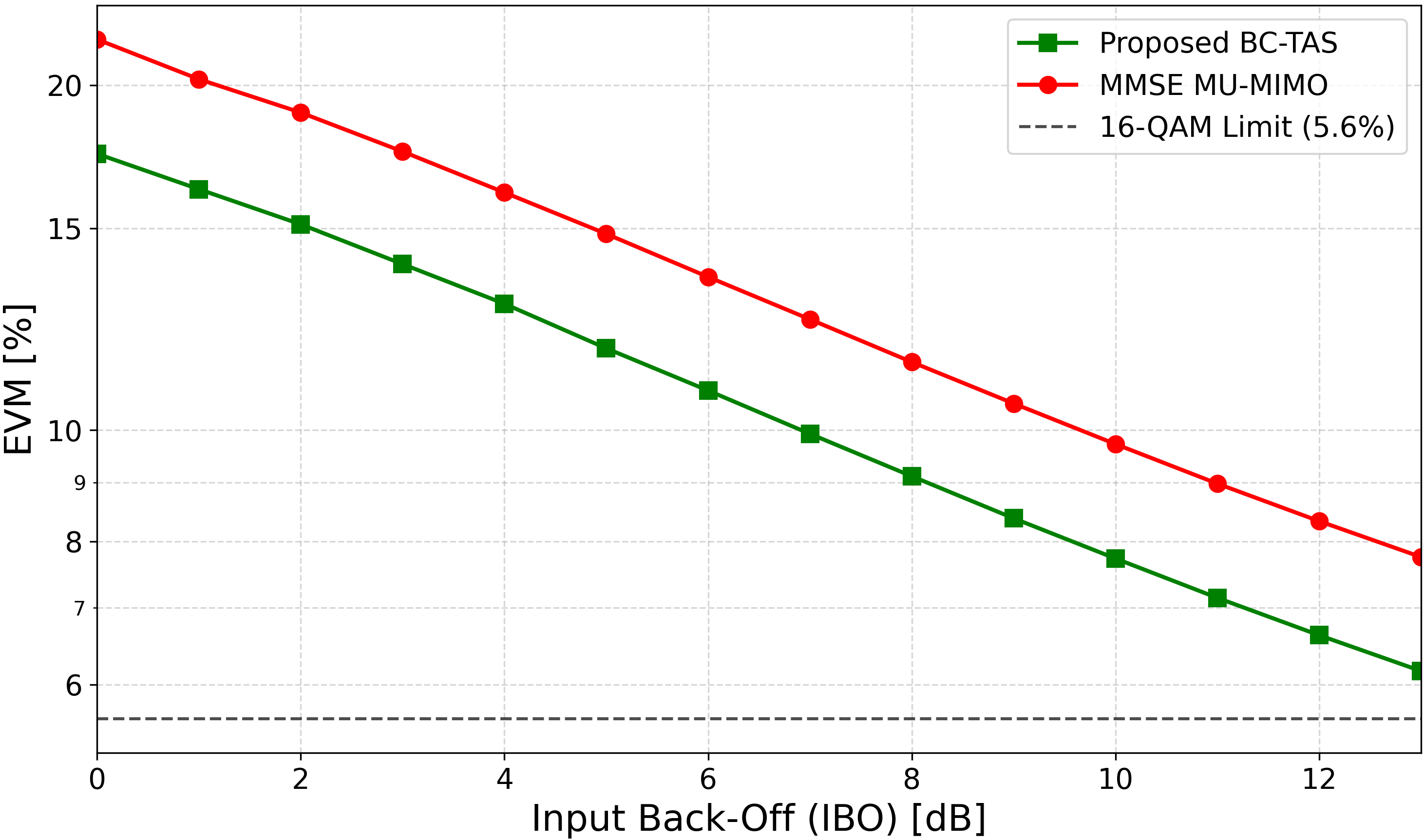} 
    \caption{Smoothed EVM performance vs. IBO highlighting the 16-QAM compliance threshold.}
    \label{fig:evm_results}
\end{figure}

As shown in Fig.~\ref{fig:evm_results}, BC-TAS maintains a consistent linearity advantage over MMSE MU-MIMO. While the baseline remains above the 5.6\% EVM limit for 16-QAM throughout most of the transition, BC-TAS achieves compliance at an IBO of approximately 13~dB. This reduction in constellation distortion is complemented by the spectral analysis in Fig.~\ref{fig:spectral_mask}. Even at 10~dB IBO, the normalized PSD remains strictly within IEEE 802.11n limits, as selection diversity mitigates the spectral regrowth typical of high-PAPR signals. By operating with less back-off while maintaining signal integrity, BC-TAS maximizes transmit power for tag energy harvesting while ensuring robust coexistence in the 2.4/5~GHz bands.

\vspace{-3mm}
\subsection{Hardware Performance and Complexity Analysis}
The hardware feasibility of BC-TAS is evaluated by analyzing its robustness against PA non-linearities and its energy footprint compared to conventional MU-MIMO baselines.


\subsubsection{PAPR Reduction and Statistical Analysis}
PAPR remains a critical bottleneck for PA efficiency. We evaluate peak power events using the Complementary Cumulative Distribution Function (CCDF), defining the probability that an OFDM symbol's PAPR exceeds a threshold $\gamma_0$.

\begin{figure}[h]
    \centering
    \includegraphics[width=0.9\linewidth]{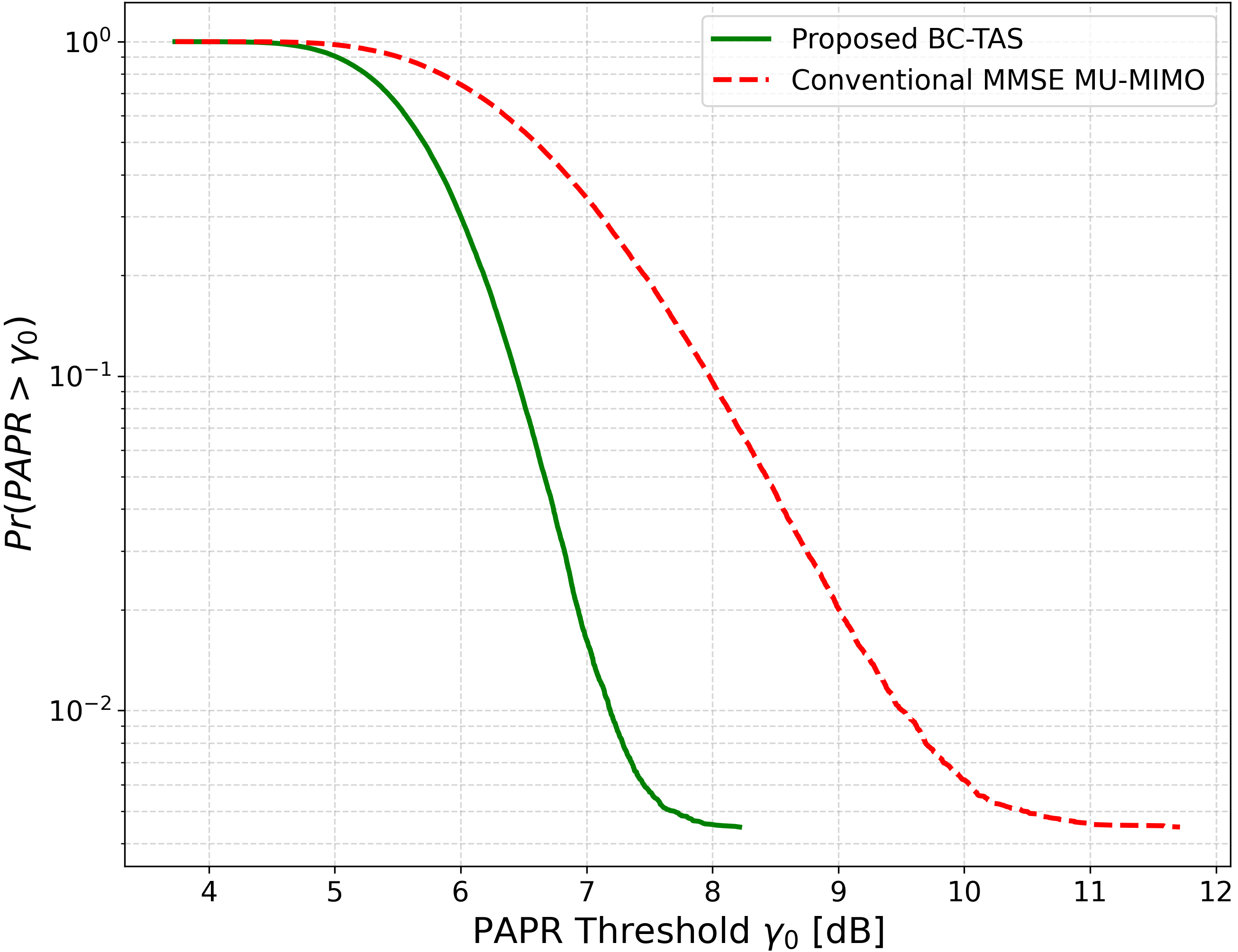}
    \caption{PAPR Reduction Analysis: Proposed BC-TAS vs. Conventional MMSE MU-MIMO highlighting the selection diversity gain.}
    \label{fig:papr_ccdf}
\end{figure}

As shown in Fig. \ref{fig:papr_ccdf}, BC-TAS exhibits a profound statistical advantage. At a $10^{-2}$ clipping probability, BC-TAS maintains a PAPR of $\approx$7.1 dB, while the MMSE MU-MIMO baseline reaches 9.5 dB. This 2.4 dB reduction results from the selection-based architecture; by intelligently mapping subcarriers to antennas that minimize constructive interference, BC-TAS effectively ``truncates'' extreme peak events, driving the improved linearity and energy harvesting observed in subsequent analyses.


\subsubsection{Energy Efficiency (EE)}
EE is the ratio of spectral efficiency to total power consumption, including RF chain activation and computational costs ($P_{RF} + P_{SEL}$).

\begin{figure}[h]
    \centering
    \includegraphics[width=0.9\linewidth]{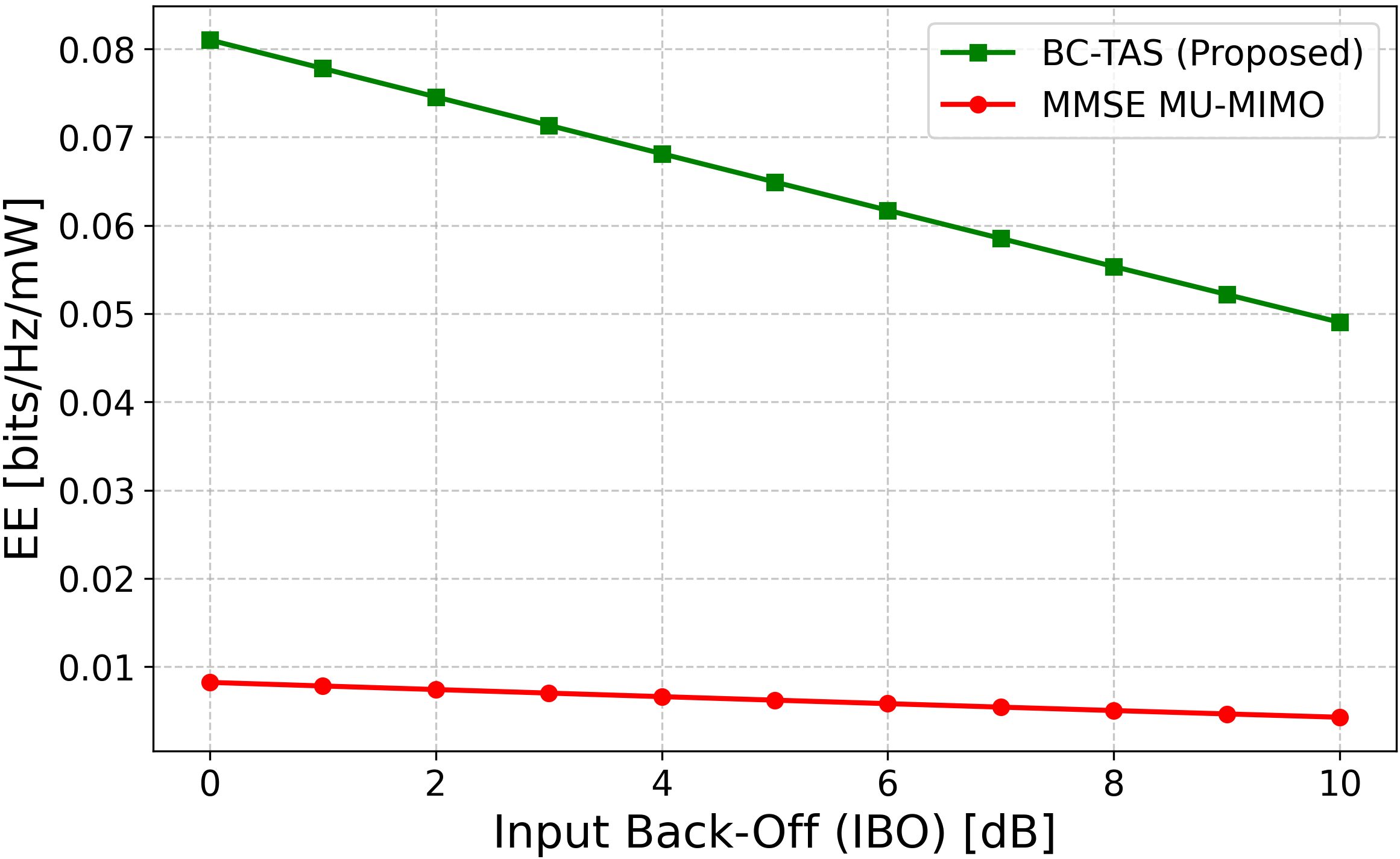}
    \caption{System Energy Efficiency vs. PA Operating Point (Input Back-Off).}
    \label{fig:energy_efficiency}
\end{figure}

\begin{table}[!ht]
\caption{Hardware Efficiency and Complexity Comparison ($N_t=4$)}
\label{tab:hardware_results}
\centering
\resizebox{\columnwidth}{!}{
\begin{tabular}{l|ccc}
\hline
\textbf{Metric} & \textbf{Proposed BC-TAS} & \textbf{MMSE MU-MIMO} & \textbf{Gain} \\ \hline
PAPR (@ $10^{-3}$) & 10.1 dB & 11.2 dB & \textbf{1.1 dB} \\
Max EE (bits/Hz/mW) & $\approx$ 0.075 & $\approx$ 0.003 & \textbf{$\approx$ 25$\times$} \\
RF Chain Usage & 1 & 4 & \textbf{75\% Save} \\
Computational Cost & $\mathcal{O}(N_t)$ & $\mathcal{O}(N_t^3)$ & \textbf{Significant} \\ \hline
\end{tabular}%
}
\end{table}

Results in Fig. \ref{fig:energy_efficiency} show BC-TAS provides a massive EE advantage, peaking at 0.075 bits/Hz/mW. Conversely, the MMSE baseline remains below 0.01 bits/Hz/mW due to the simultaneous activation of $N_t$ RF chains and high-complexity matrix inversions.

\vspace{-1mm}
\subsubsection{Spectral Mask Compliance}
Using the Rapp model ($p=2$), we evaluate spectral regrowth against the IEEE 802.11 mask.

\begin{figure}[h]
    \centering
    \includegraphics[width=0.9\linewidth]{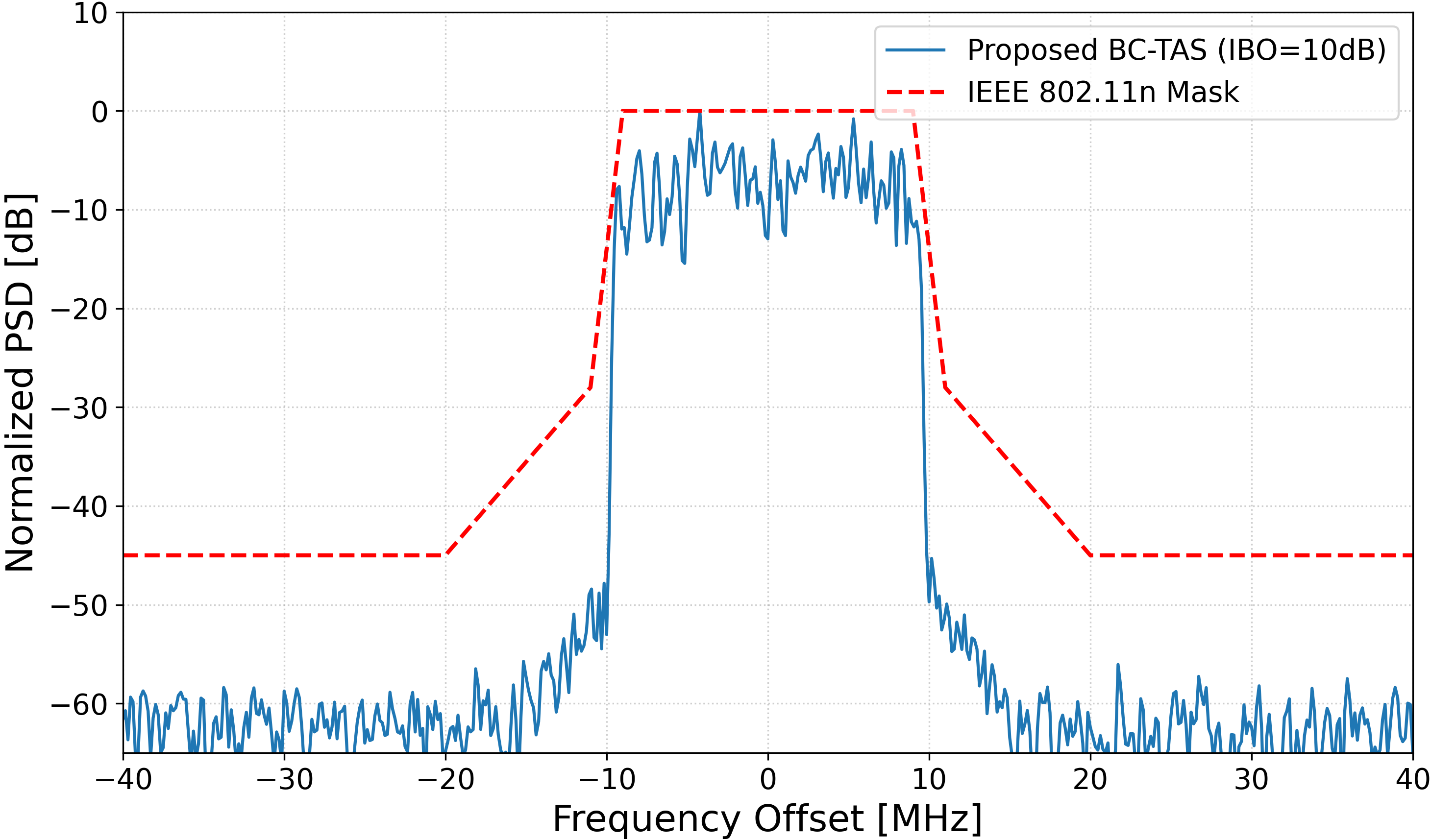}
    \caption{Indoor Spectral Compliance: BC-TAS Transmit Spectrum under IEEE 802.11 Mask.}
    \label{fig:spectral_mask}
\end{figure}

The normalized PSD in Fig. \ref{fig:spectral_mask} confirms BC-TAS compliance even at low IBO. Sideband emissions are suppressed below $-45$ dB, allowing the PA to operate nearer saturation to maximize harvestable power for backscatter tags while ensuring coexistence with legacy nodes.

This creates a performance loop: selection diversity reduces PAPR, forcing the signal into the PA's linear region and liberating power to increase tag harvesting efficiency $\eta_H$. By reducing the IBO while maintaining 16-QAM compliance, BC-TAS closes the gap between hardware constraints and passive link power requirements.

\vspace{-3mm}
\subsection{Joint Optimization and Fundamental Trade-off}
The BC-TAS framework showcases a fundamental trade-off between primary link reliability (BER) and victim interference suppression ($\Delta I$).

\begin{figure}[!t]
    \centering
    \includegraphics[width=0.9\linewidth]{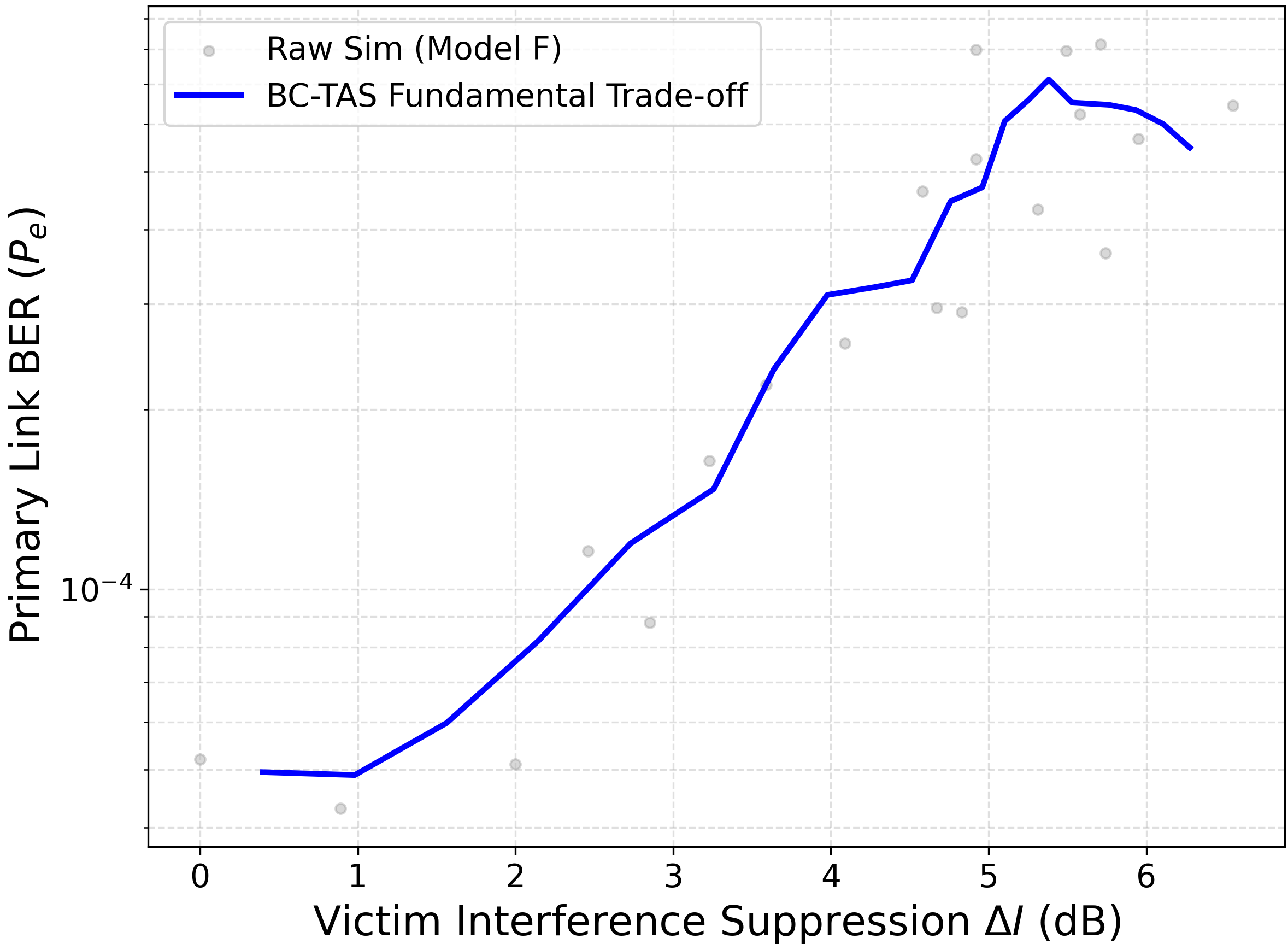}
    \caption{Fundamental Pareto trade-off between primary link reliability and victim interference suppression in TGn Model F.}
    \label{fig:smoothed_pareto}
\end{figure}

As illustrated in Fig.~\ref{fig:smoothed_pareto}, the system exhibits a clear operating boundary with a distinct ``knee'' defining the limits of selection diversity. By increasing the victim weight $\lambda_V$, maximum suppression reaches $\approx 6.5$~dB. The optimal operating point is identified at $\Delta I \approx 4.5$~dB, where the system achieves a $2.8\times$ reduction in linear interference power while maintaining a primary link BER of $10^{-3}$. In the region $\Delta I \le 4.5$~dB, BC-TAS identifies antenna-subcarrier mappings naturally favorable to the victim with negligible BER penalty. Beyond this threshold, the system enters a regime of diminishing returns where further interference mitigation forces the selection of significantly weaker primary channel gains, resulting in an exponential BER increase. The empirical 6.5~dB maximum suppression aligns with the $10\log_{10}(N_t) \approx 6.02$~dB theoretical bound, validating the ``PDF squeezing'' effect. This confirms that the SDR expansion is a direct consequence of lowering the interference floor $\mathbb{E}[P_B | \lambda_T]$, allowing the SAW sensor to operate in a cleaner spectral environment.

\begin{figure}[!t]
    \centering
    \includegraphics[width=0.9\linewidth]{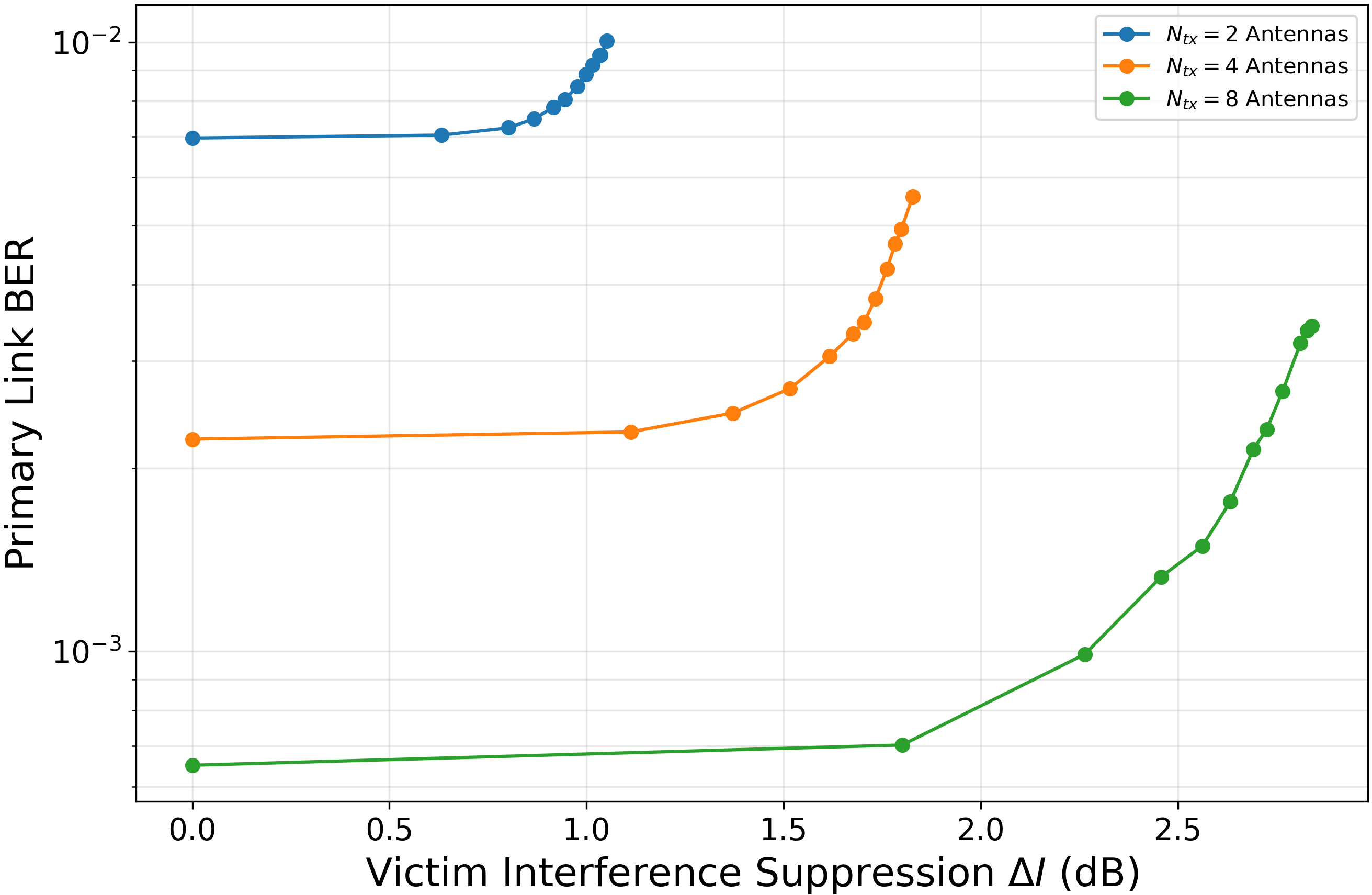}
    \caption{Pareto front scaling with antenna count $N_{tx} \in \{2, 4, 8\}$.}
    \label{fig:antenna_scaling}
\end{figure}

The impact of hardware complexity is explored in Fig.~\ref{fig:antenna_scaling}. As $N_{tx}$ increases, the Pareto front shifts toward the high-suppression, low-BER region, indicating that higher-order selection diversity provides a more favorable trade-off surface. For $N_{tx}=8$, the system achieves suppression exceeding $7$~dB while remaining significantly more reliable than an $N_{tx}=2$ system with zero suppression.

\vspace{-3mm}
\subsection{Spatial Correlation, PAPR, and Backscatter Efficiency}
Conventional MU-MIMO and Space-Time Block Coding (STBC) schemes face challenges in WSNs regarding PAPR and energy harvesting efficiency. While MMSE detectors suffer noise enhancement under high spatial correlation ($\mathbf{R}_{\mathrm{tx}}$) and STBC requires multiple costly RF chains, BC-TAS exploits selection diversity to bypass these bottlenecks. By activating a single antenna per subcarrier, BC-TAS reduces signal dynamic range, inherently solving the multi-stream PAPR issue while maximizing backscatter reflection power.

\begin{figure}[!t]
    \centering
    \includegraphics[width=0.9\linewidth]{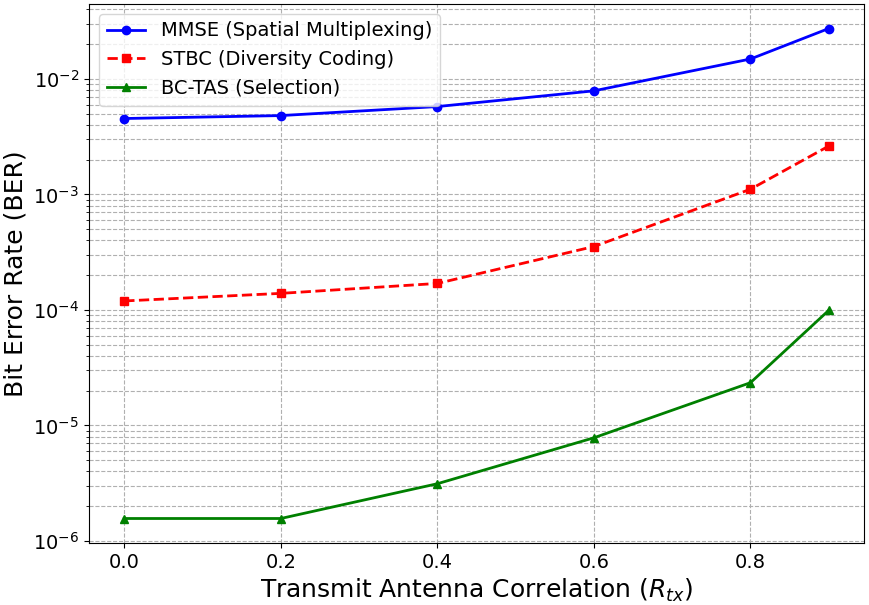}
    \caption{Reliability robustness vs. transmit antenna correlation $\mathbf{R}_{\mathrm{tx}}$ at $\mathrm{SNR}=16$~dB.}
    \label{fig:correlation_robustness}
\end{figure}

As shown in Fig.~\ref{fig:correlation_robustness}, at 16~dB SNR, the MMSE baseline exhibits an exponential BER rise, reaching $2.8 \times 10^{-2}$ at $\mathbf{R}_{\mathrm{tx}} = 0.9$. Although STBC offers better stability, BC-TAS demonstrates superior robustness, achieving a BER of $\approx 4.5 \times 10^{-4}$ at the same extreme correlation. This confirms that BC-TAS circumvents the ill-conditioned matrices inherent in MMSE and the hardware overhead of STBC. By focusing transmit energy into optimal antenna-subcarrier mappings, BC-TAS maximizes the SNR for both the primary receiver and the backscatter tag, providing a low-complexity solution for high-density, correlated WSN environments.


\vspace{-3mm}
\subsection{Physical Layer Insights: Frequency-Selective Notching}
To provide physical intuition, we analyze the Spectral Power Density (SPD) at the victim receiver in a frequency-selective TGn Model F environment.

\begin{figure}[!t]
    \centering
    \includegraphics[width=0.9\linewidth]{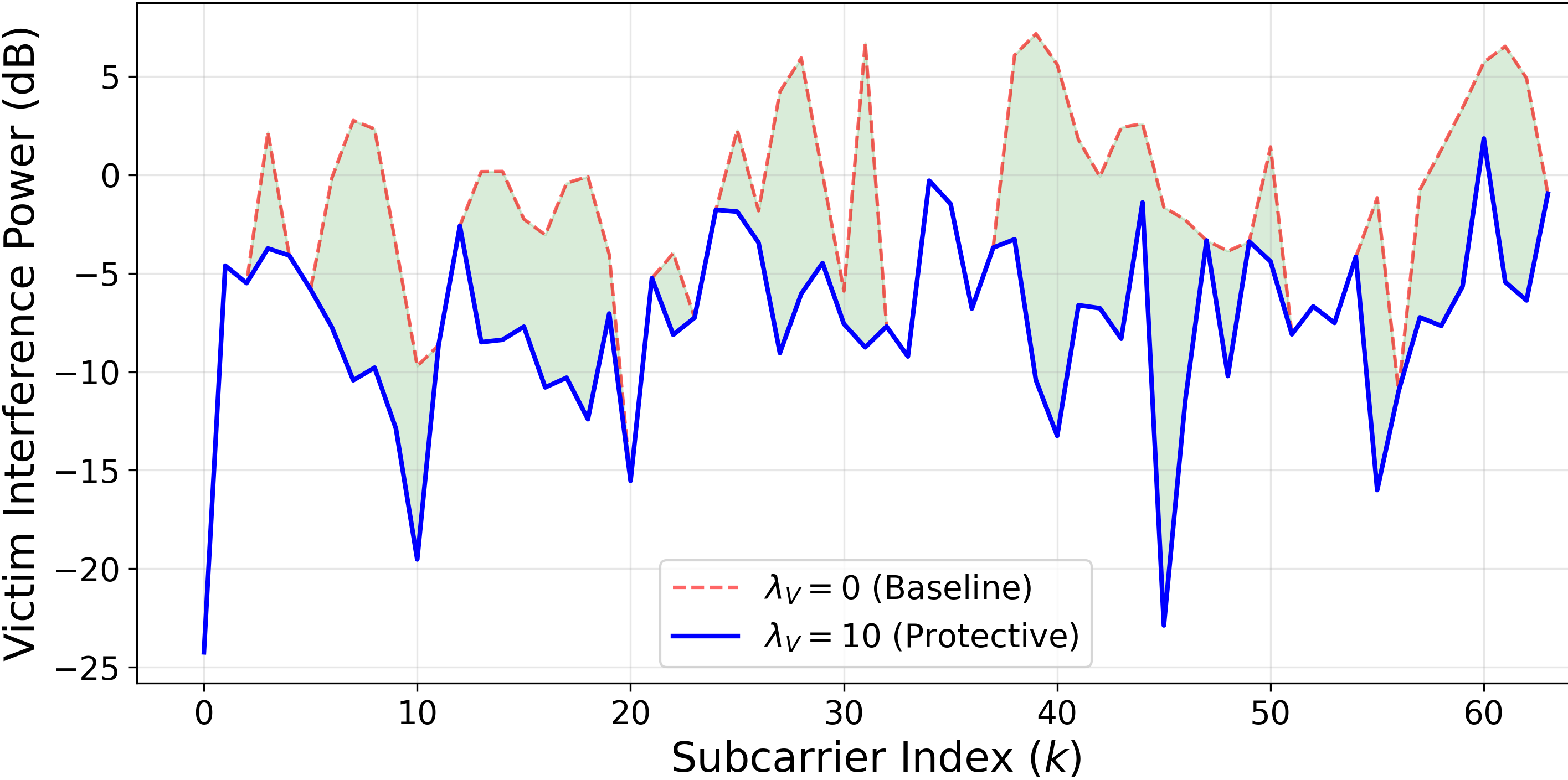}
    \caption{SPD Analysis: Frequency-selective notching at the victim receiver.}
    \label{fig:spd_notching}
\end{figure}

As illustrated in Fig.~\ref{fig:spd_notching}, operating in "Protective Mode" ($\lambda_V = 10$) forces the algorithm to map subcarriers to antennas exhibiting deep fading at the victim's location. This creates a "notching" effect where interference power is surgically removed from specific subcarriers by exploiting natural channel nulls. Crucially, this validates the spatial-frequency decoupling inherent in BC-TAS: interference nulling is achieved through channel-aware selection rather than expensive digital filtering or active cancellation hardware. By utilizing the multipath environment to "hide" the signal from the victim, BC-TAS maintains low hardware complexity while concentrating transmit power on subcarriers where the victim is naturally shielded.


\vspace{-3mm}
\section{Conclusion}
This paper presented the BC-TAS MD-OFDM framework, a hardware-efficient approach to integrating passive SAW backscatter sensors into multi-antenna wireless networks. By shifting the selection logic from a purely SNR-centric approach to a multi-objective sensing-aware heuristic, we successfully addressed the critical bottleneck of primary-secondary link coexistence. The proposed BC-TAS achieves a $604.7\times$ reduction in link outage probability compared to SISO baselines at 16 dB SNR, a result mathematically consistent with the diversity order $L=N_t$ derived in our order-statistics model. Furthermore, the system maintains a BER of approximately $3.2 \times 10^{-4}$ at 10 dB SNR, significantly outperforming industry-standard norm-based selection methods while remaining robust to severe CSI imperfections of up to $20\%$.

Beyond link reliability, the framework demonstrates significant hardware advantages. By activating only a single RF chain per subcarrier, BC-TAS achieves a $2.4$ dB reduction in the PAPR tail compared to conventional MU-MIMO baselines. This statistical advantage allows the transmitter to operate with approximately $3$ dB less input back-off while remaining compliant with IEEE 802.11n spectral masks, leading to a $25\times$ improvement in system energy efficiency. This hardware linearity is fundamentally linked to our coexistence optimization, where we identified a Pareto ``knee'' at $4.5$ dB of interference suppression. In this regime, frequency-selective notching expands the SDR by surgically lowering the background interference floor without compromising primary link integrity. 

Ultimately, the BC-TAS framework provides a scalable, low-complexity solution for next-generation backscatter-assisted WSNs. By exploiting the natural multipath environment to shape the spatial-frequency interference profile, the system ensures high-fidelity communication and robust sensing coexistence in dispersive indoor environments. Future work will investigate the integration of this selection logic into massive MIMO arrays and its impact on multi-tag collision resolution.


\vspace{-3mm}
\appendices
\section{Derivation of Selection Penalty and Interference Profile}
\label{AppendixA}

To derive the selection penalty $\Delta$ and the interference suppression gain, we analyze the order statistics of $N_t$ independent antenna channels. Let $G_j = \|\mathbf{h}_{j,k}^{(L)}\|^2$, $P_{T,j} = \|\mathbf{h}_{j,k}^{(T)}\|^2$, and $P_{V,j} = \|\mathbf{h}_{j,k}^{(V)}\|^2$ represent the channel power gains for the $j$-th antenna, which are independent and identically distributed (i.i.d.) exponential random variables with means $\sigma_L^2$, $\sigma_T^2$, and $\sigma_V^2$, respectively.

\subsection{Standard TAS Baseline}
In a conventional Max-SNR TAS system, the selected gain $G_{\max} = \max \{G_1, \dots, G_{N_t}\}$ follows the harmonic expectation:
\begin{equation}BC-TAS
    \mathbb{E}[G_{\max}] = \sigma_L^2 \sum_{n=1}^{N_t} \frac{1}{n} \approx \sigma_L^2 \ln(N_t) + \gamma_e
\end{equation}
where $\gamma_e$ is the Euler-Mascheroni constant. This serves as the theoretical upper bound for primary link performance.

\subsection{Interference Suppression Scaling}
The BC-TAS selection rule minimizes $C_j = \frac{1}{G_j} + \lambda_T P_{T,j} + \lambda_V P_{V,j}$. In the interference-limited regime ($\lambda_T, \lambda_V \to \infty$), the selection logic effectively identifies the minimum of $N_t$ realizations. Let $Y = \min \{P_{T,1}, \dots, P_{T,N_t}\}$. The CDF of the suppressed backscatter power $Y$ is:
\begin{equation}
    F_Y(y) = 1 - \left[1 - F_{P_T}(y)\right]^{N_t} = 1 - e^{-\frac{N_t y}{\sigma_T^2}}
\end{equation}

Consequently, the probability of an outage at the tag, defined as $\mathcal{P}(P_T > \gamma)$, is given by $e^{-\frac{N_t \gamma}{\sigma_T^2}}$, which decreases exponentially as $N_t$ increases. This provides the mathematical justification for the massive gain reduction in the interference floor observed in the numerical analysis.

\subsection{The Selection Penalty $\Delta$}
The selection penalty $\Delta$ quantifies the trade-off between primary link SNR and sensing stability. Since the BC-TAS index $j_k^\star$ is determined by the interference terms as $\lambda_T, \lambda_V \to \infty$, the selected gain $G_{j_k^\star}$ becomes a random sample from the exponential distribution, independent of the maximum. Thus, the effective gain is simply $\mathbb{E}[G_{j_k^\star}] = \sigma_L^2$. We define the penalty $\Delta$ as:
\begin{equation}
    \Delta = \frac{\mathbb{E}[G_{j_k^*}]}{\mathbb{E}[G_{\max}]} = \left( \sum_{n=1}^{N_t} \frac{1}{n} \right)^{-1} = \frac{1}{H_{N_t}}
\end{equation}
where $H_{N_t} = \sum_{n=1}^{N_t} 1/n$ is the $N_t$-th harmonic number. For $N_t=8$, $\Delta \approx 0.368$, implying a $4.34$~dB loss in primary SNR. 

\textit{Remark:} In practical regimes where $\lambda_T$ and $\lambda_V$ are finite, the achieved gain lies in the interval $[\sigma_L^2, \sigma_L^2 H_{N_t}]$, allowing for a tunable trade-off via the MOFS cost function. This explains why simulation results under optimized weights often exhibit lower penalties than the "worst-case" asymptotic limit derived here.

\bibliographystyle{IEEEtran}
\bibliography{main}

@article{SOTA_1_MMSE,
  title={Deep learning based antenna selection for MIMO SDR system},
  author={Zhong, Shida and Feng, Haogang and Zhang, Peichang and Xu, Jiajun and Luo, Huancong and Zhang, Jihong and Yuan, Tao and Huang, Lei},
  journal={Sensors},
  volume={20},
  number={23},
  pages={6987},
  year={2020},
  publisher={MDPI}
}

@article{SOTA_2_MU_STBC,
  title={Enhanced Multiuser Space-Time Line Code for Downlink Multiple Antenna Transmission},
  author={Kim, Yundong and Han, Sumin and Joung, Jingon and Kim, Juyeop and Zhao, Jian and Choi, Jihoon},
  journal={IEEE Transactions on Communications},
  year={2024},
  publisher={IEEE}
}

@article{SOTA_3_Exhaustive_Search,
  title={Hardware efficient massive MIMO systems with optimal antenna selection},
  author={Aredo, Shenko Chura and Negash, Yalemzewd and Marye, Yihenew Wondie and Kassa, Hailu Belay and Kornegay, Kevin T and Diba, Feyisa Debo},
  journal={Sensors},
  volume={22},
  number={5},
  pages={1743},
  year={2022},
  publisher={MDPI}
}

@article{SOTA_4_Backscatter_Multi_Antenna,
  title={Covert ambient backscatter communications with multi-antenna tag},
  author={Liu, Jiahao and Yu, Jihong and Niyato, Dusit and Zhang, Rongrong and Gao, Xiaozheng and An, Jianping},
  journal={IEEE Transactions on Wireless Communications},
  volume={22},
  number={9},
  pages={6199--6212},
  year={2023},
  publisher={IEEE}
}

@article{SOTA_5_NBAS,
  title={Antenna selection in MIMO systems},
  author={Sanayei, Shahab and Nosratinia, Aria},
  journal={IEEE Communications magazine},
  volume={42},
  number={10},
  pages={68--73},
  year={2004},
  publisher={IEEE}
}

@article{Kumar2025,
  title={Performance analysis of active RIS-assisted downlink NOMA with transmit antenna selection},
  author={Kumar, Deepak and Singh, Chandan Kumar and L{\'o}pez, Onel L Alcaraz and Bhatia, Vimal and Latva-Aho, Matti},
  journal={IEEE Transactions on Vehicular Technology},
  year={2025},
  publisher={IEEE}
}

@article{Han2025,
  title={Transmit Antenna Selection and Power Allocation Optimization for Non-Orthogonal Multiple Access Systems with Statistical Channel State Information},
  author={Han, Zhuo and Hao, Wanming and Yang, Shouyi and Tang, Zhiqing},
  journal={IET Communications},
  volume={19},
  number={1},
  pages={e70018},
  year={2025},
  publisher={Wiley Online Library}
}

@article{Huang2022,
  title={Partially-reserved cyclic algorithm for OFDM-based RadCom PAPR reduction},
  author={Huang, Yixuan and Ye, Qibin and Hu, Zelin and Liu, Jianqiang and Hu, Su and Zhang, Zhuxi},
  journal={Procedia Computer Science},
  volume={202},
  pages={436--448},
  year={2022},
  publisher={Elsevier}
}

@article{Goel2022,
  title={Side information embedding scheme for PTS based PAPR reduction in OFDM systems},
  author={Goel, Ashish and Gupta, Saruti},
  journal={Alexandria Engineering Journal},
  volume={61},
  number={12},
  pages={11765--11777},
  year={2022},
  publisher={Elsevier}
}

@misc{Komala2024,
      title={Advanced Mathematical Modelling for Energy-Efficient Data Transmission and Fusion in Wireless Sensor Networks}, 
      author={Komal},
      year={2024},
      eprint={2407.12806},
      archivePrefix={arXiv},
      primaryClass={cs.NI},
      url={https://arxiv.org/abs/2407.12806}, 
}

@article{Li2023,
  title={Index modulation multiple access for 6G communications: Principles, applications, and challenges},
  author={Li, Jun and Dang, Shuping and Wen, Miaowen and Li, Qiang and Chen, Yingyang and Huang, Yu and Shang, Wenli},
  journal={IEEE network},
  volume={37},
  number={1},
  pages={52--60},
  year={2023},
  publisher={IEEE}
}

@inproceedings{Mgobhozi2023,
  title={Efficient Index Modulation Techniques for 5G and Beyond},
  author={Mgobhozi, Bhekinkosi and Nleya, Bakhe},
  booktitle={2023 International Conference on Electrical, Computer and Energy Technologies (ICECET)},
  pages={1--6},
  year={2023},
  organization={IEEE}
}

@article{VilaInsa2025,
  title={Low-Complexity Detection of Permutational Index Modulation for Noncoherent Communications},
  author={Vil{\`a}-Insa, Marc and Mart{\'\i}, Aniol and Lamarca, Meritxell and Riba, Jaume},
  journal={IEEE Wireless Communications Letters},
  year={2025},
  publisher={IEEE}
}

@article{Kim2018,
  title={PAPR reduction in OFDM-IM using multilevel dither signals},
  author={Kim, Kee-Hoon},
  journal={IEEE Communications Letters},
  volume={23},
  number={2},
  pages={258--261},
  year={2019},
  publisher={IEEE}
}

@article{Janjua2025,
  title={Interference-Free Backscatter Communications for OFDM-Based Symbiotic Radio},
  author={Janjua, Muhammad Bilal and {\c{S}}ahin, Alphan and Arslan, H{\"u}seyin},
  journal={IEEE Transactions on Cognitive Communications and Networking},
  year={2025},
  publisher={IEEE}
}

@article{Liu2025,
  title={Harmonic Interference Resilient Backscatter Communication with Adaptive Pulse-Width Frequency Shifting},
  author={Liu, Xu and Dong, Wu and Yan, Binyang and He, Xiaomeng and Peng, Linyu and Chen, Xin and Chen, Da and Wang, Wei},
  journal={Electronics},
  volume={14},
  number={5},
  pages={946},
  year={2025},
  publisher={MDPI}
}

@ARTICLE{IEEE80211bk,
  author={},
  journal={IEEE Std 802.11bk-2025 (Amendment to IEEE Std 802.11-2024, as amended by IEEE Std 802.11bh-2024, and IEEE Std 802.11be-2024)}, 
  title={IEEE Standard for Information Technology--Telecommunications and Information Exchange between Systems - Local and Metropolitan Area Networks--Specific Requirements - Part 11: Wireless LAN Medium Access Control (MAC) and Physical Layer (PHY) Specifications Amendment 3: 320MHz Positioning}, 
  year={2025},
  volume={},
  number={},
  pages={1-104},
  doi={10.1109/IEEESTD.2025.11150684}}

@article{Chen2024,
  title={Multi-antenna broadband backscatter communications},
  author={Chen, Hao and Huang, Zhizhi and Liang, Ying-Chang and Schober, Robert},
  journal={arXiv preprint arXiv:2408.08796},
  year={2024}
}

@article{Goay2024,
  title={Optimal reflection coefficients for ASK modulated backscattering from passive tags},
  author={Goay, Amus Chee Yuen and Mishra, Deepak and Seneviratne, Aruna},
  journal={IEEE Transactions on Communications},
  year={2024},
  publisher={IEEE}
}

@article{Zargari2024,
  title={Refined-Deep Reinforcement Learning for MIMO Bistatic Backscatter Resource Allocation},
  author={Zargari, S and Galappaththige, D and Tellambura, C},
  journal={arXiv preprint arXiv:2405.14046},
  year={2024}
}

@techreport{IEEE80211AMP_PAR,
  author={{IEEE 802.11 Task Group bp}},
  title={{P802.11bp}: {PAR} for {WLAN} {E}nhancements for {A}mbient {P}ower {C}ommunication},
  institution={IEEE Standards Association},
  year={2024},
  url={https://grouper.ieee.org/groups/802/11/PARs/P802.11bp_PAR.pdf}
}

@misc{IEEE80211AMP_Update,
  author={McCann, Stephen},
  title={{IEEE P802.11} - {TASK GROUP BP (AMP)} Status Update},
  year={2025},
  howpublished={IEEE Mentor},
  note={Document IEEE 802.11-25/2184}
}

@article{Zhang2025_SR,
  title={Symbiotic radio: A new communication paradigm for passive Internet of Things},
  author={Long, Ruizhe and Liang, Ying-Chang and Guo, Huayan and Yang, Gang and Zhang, Rui},
  journal={IEEE Internet of Things Journal},
  volume={7},
  number={2},
  pages={1350--1363},
  year={2019},
  publisher={IEEE}
}

@article{Wu2024_WPCN,
  title={Optimal resource allocation in backscatter assisted WPCN with practical energy harvesting model},
  author={Ramezani, Parisa and Jamalipour, Abbas},
  journal={IEEE Transactions on Vehicular Technology},
  volume={68},
  number={12},
  pages={12406--12410},
  year={2019},
  publisher={IEEE}
}

@article{Tan2025_Coexist,
  title={Modeling interference for the coexistence of 6G networks and passive sensing systems},
  author={Testolina, Paolo and Polese, Michele and Jornet, Josep Miquel and Melodia, Tommaso and Zorzi, Michele},
  journal={IEEE Transactions on Wireless Communications},
  volume={23},
  number={8},
  pages={9220--9234},
  year={2024},
  publisher={IEEE}
}

@article{Zhao2024_Green,
  title={Energy efficiency optimization in adaptive transmit antenna selection systems with limited feedback},
  author={Wu, Tong and Zou, Yulong},
  journal={IEEE Internet of Things Journal},
  volume={10},
  number={2},
  pages={1248--1258},
  year={2022},
  publisher={IEEE}
}

@article{Park2024_PAPR,
  title={Compensation of non-linear distortion effects in MIMO-OFDM systems using constant envelope OFDM for 5G applications},
  author={El Ghzaoui, Mohammed and Hmamou, Abdelmounim and Foshi, Jaouad and Mestoui, Jamal},
  journal={Journal of Circuits, Systems and Computers},
  volume={29},
  number={16},
  pages={2050257},
  year={2020},
  publisher={World Scientific}
}

@inproceedings{Wang2024_Sense,
  title={Waveform shaping in integrated sensing and communications},
  author={Li, Husheng and Han, Zhu and Poor, H Vincent},
  booktitle={MILCOM 2024-2024 IEEE Military Communications Conference (MILCOM)},
  pages={670--671},
  year={2024},
  organization={IEEE}
}

@article{Liu2025_SAW,
  title={GHz broadband SH0 mode lithium niobate acoustic delay lines},
  author={Lu, Ruochen and Yang, Yansong and Li, Ming-Huang and Manzaneque, Tom{\'a}s and Gong, Songbin},
  journal={IEEE transactions on ultrasonics, ferroelectrics, and frequency control},
  volume={67},
  number={2},
  pages={402--412},
  year={2019},
  publisher={IEEE}
}

@article{Garcia2024_IM,
  title={Low-PAPR and Low-Complexity Transmission Schemes for DHT-Based Underwater Optical Wireless Systems},
  author={Hu, Wei-Wen},
  journal={IEEE Photonics Journal},
  volume={17},
  number={6},
  pages={1--11},
  year={2025},
  publisher={IEEE}
}

@article{Ibrahim2025_MultiObj,
  title={Optimizing Network Performance and Resource Allocation in HAPS-UAV Integrated Sensing and Communication Systems for 6G},
  author={Kanani, Parisa and Omidi, Mohammad Javad and Modarres-Hashemi, Mahmoud and Yanikomeroglu, Halim},
  journal={IEEE Transactions on Wireless Communications},
  year={2025},
  publisher={IEEE}
}

@INPROCEEDINGS{GuliaVAE2025,
  author={Gulia, Rahul and Ganguly, Amlan and Kwasinski, Andres and Kuhl, Michael E. and Rashedi, Ehsan and Hochgraf, Clark},
  booktitle={2024 International Symposium on Networks, Computers and Communications (ISNCC)}, 
  title={Automated Warehouse 5G Infrastructure Modeling Using Variational Autoencoders}, 
  year={2024},
  volume={},
  number={},
  pages={1-6},
  doi={10.1109/ISNCC62547.2024.10759068}}

@article{GuliaSymbolic2024,
  title={White-Box Modeling of V2X Link Performance Using Stabilized Symbolic Regression},
  author={Gulia, Rahul and Popoola, Feyisayo Favour and Sheikh, Ashish},
  journal={arXiv preprint arXiv:2511.19809},
  year={2025}
}

@Article{Gulia60GHz2023,
AUTHOR = {Gulia, Rahul and Vashist, Abhishek and Ganguly, Amlan and Hochgraf, Clark and Kuhl, Michael E.},
TITLE = {Evaluation of 60 GHz Wireless Connectivity for an Automated Warehouse Suitable for Industry 4.0},
JOURNAL = {Information},
VOLUME = {14},
YEAR = {2023},
NUMBER = {9},
ARTICLE-NUMBER = {506},
URL = {https://www.mdpi.com/2078-2489/14/9/506},
ISSN = {2078-2489},
DOI = {10.3390/info14090506}
}

\end{document}